\documentclass[journal]{IEEEtran}

\usepackage{color}
\usepackage{enumerate}
\usepackage{setspace}
\usepackage{bm}
\usepackage{amsthm}
\usepackage[cmex10]{amsmath}
\usepackage{dsfont}
\usepackage{colortbl}
\usepackage{stfloats}
\usepackage{amssymb}
\usepackage{algorithm}
\usepackage{algorithmic}
\usepackage{multirow}
\usepackage{xcolor}
\usepackage{graphicx}
\usepackage{stmaryrd}
\usepackage{cite}
\usepackage{multirow}
\usepackage{multicol}
\usepackage{booktabs}
\usepackage{longtable}
\usepackage{subfigure}

\theoremstyle{definition}

\newtheorem{prop}{Proposition}

\renewcommand{\algorithmicrequire}{ \textbf{Input:}}      
\renewcommand{\algorithmicensure}{ \textbf{Output:}}

\usepackage{epstopdf}
\IEEEoverridecommandlockouts

\begin{document}
%
\title{\Huge Coordinated RSMA for Integrated Sensing and Communication in Emergency UAV Systems}
\author{
	\IEEEauthorblockN{Binghan Yao, Ruoguang Li, \IEEEmembership{Member, IEEE}, Yingyang Chen,  \IEEEmembership{Senior Member, IEEE}, and Li Wang,~\IEEEmembership{Senior Member,~IEEE}\\}
	\thanks{
		This work was supported in part by the National Natural Science Foundation of China under Grant U2066201, 62301157, and 62171054, in part by the Natural Science Foundation of Jiangsu Province of China under Project BK20230823, in part by the Fundamental Research Funds for the Central Universities under Grant 24820232023YQTD01, in part by the Double First-Class Interdisciplinary Team Project Funds under Grant 2023SYLTD06, and in part by the Guangdong Basic and Applied Basic Research Project under Grant 2024B1515020002, 2023A1515012892, and 2021B1515120067. \emph{(Corresponding author: Li Wang.)}
		
		Binghan Yao and Li Wang are with the School of Computer Science (National Pilot Software Engineering School), Beijing University of Posts and Telecommunications, Beijing 100876, China (e-mail: 
		\{yaobinghan, liwang\}@bupt.edu.cn).
		
		Ruoguang Li is with the College of Information Science and Engineering, Hohai University, Changzhou 213200, China (e-mail: ruoguangli@hhu.edu.cn).
		
		Yingyang Chen is with the College of Information Science and Technology, Jinan University, Guangzhou 510632, China (e-mail: chenyy@jnu.edu.cn).
	}
}
%
%
%
\maketitle

\begin{abstract}
The destruction of terrestrial infrastructures and wireless resource-scarcity in disaster scenarios pose challenges for providing a prompt and reliable communication and sensing (C\&S) services to search and rescue operations. Recently, unmanned aerial vehicle (UAV)-enabled integrated sensing and communication (ISAC) is emerging as a promising technique for achieving robust and rapid emergency response capabilities. Such a novel framework offers high-quality and cost-efficient C\&S services due to the intrinsic flexibility and mobility of UAVs. In parallel, rate-splitting multiple access (RSMA) is able to achieve a tailor-made communication by splitting the messages into private and common parts with adjustable rates, making it suitable for on-demand data transmission in disaster scenarios. In this paper, we propose a coordinated RSMA for integrated sensing and communication (CoRSMA-ISAC) scheme in emergency UAV system to facilitate search and rescue operations, where a number of ISAC UAVs simultaneously communicate with multiple communication survivors (CSs) and detect a potentially trapped survivor (TS) in a coordinated manner. Towards this end, an optimization problem is formulated to maximize the weighted sum rate (WSR) of the system, subject to the sensing signal-to-noise ratio (SNR) requirement. In order to solve the formulated non-convex problem, we first decompose it into three subproblems, i.e., UAV-CS association, UAV deployment, as well as beamforming optimization and rate allocation. Subsequently, we introduce an iterative optimization approach leveraging K-Means, successive convex approximation (SCA), and semi-definite relaxation (SDR) algorithms to reframe the subproblems into a more tractable form and efficiently solve them. Simulation results demonstrate that the proposed CoRSMA-ISAC scheme is superior to conventional space division multiple access (SDMA), non-orthogonal multiple access (NOMA), and orthogonal multiple access (OMA) in terms of both communication and sensing performance.

\end{abstract}
\section{Introduction}
Disaster scenarios generally require a prompt emergency response, which demands reliable and uninterrupted wireless communication and sensing (C\&S) services to facilitate the transmission of rescue tasks and the detection of trapped survivors (TSs) \cite{D2019,EKMarkakis2017,WL2020edge}. However, the conventional terrestrial infrastructure is often out of work after disasters due to the severe destruction. In particular, the obstacles such as mountains and buildings may block the line-of-sight (LoS) links between the transmitters and receivers, resulting in a seriously degraded performance of on-site and timely C\&S services. Driven by the flexible mobility and on-demand connectivity, unmanned aerial vehicle (UAV) is envisioned as an aerial multi-functional platform that can be harnessed for various applications in emergency events \cite{LGupta2016,YZeng2019,Zhao2019}, such as disaster warnings broadcasting, medical supplies delivery, survivor/environemt status monitoring, etc. Specifically, by exploiting strong LoS links, emergency UAV system can provide enhanced wireless C\&S services from sky for rescuers and survivors with the extended coverage, higher capacity, and quicker deployment in rescue operations \cite{WL2023}. Extensive researches have mainly focused on the emergency UAV-enabled C\&S network for disaster scenarios from the perspectives of framework design, trajectory design, resource allocation, and multi-UAV deployment, etc \cite{ZYao2021,TDo-Duy2021,SWu2022,NLin2022,WL2021Joint3D}. 

However, the separate deployment of C\&S services inevitably incurs a heavy payload for emergency equipment, especially for that with limited size and available battery power, such as UAV \cite{WL2017}. Therefore, by introducing integrated sensing and communication (ISAC) technology to unify wireless communication and radar sensing functionalities into the emergency UAV system,  a  better adaption to search and rescue operations can be achieved in disaster scenarios with higher resource efficiency \cite{Fliu2022,JAZhang2022,JYang2023,ZWei2023}. Recently, emergency UAV-enabled ISAC systems have attracted attention in both academia and industry, cost-efficiently improving C\&S performance with reused resource and payload reduction\cite{JMu2023,LWang2023,WL2023Joint,KMeng2023,KMeng20232,XJing2022,KZhang2021}. For instance, \cite{KMeng2023} gave an overview of UAV-enabled ISAC system, including the basic network architecture and technical issues. The authors in \cite{KMeng20232} proposed an integrated periodic sensing and communication (IPSAC) mechanism, which flexibly provided C\&S services for multiple users equipment (UEs) and targets, and jointly optimized the UAV trajectory, user association, sensing target selection, as well as the transmit beamforming to maximize the achievable communication rate. The work \cite{XJing2022} investigated the trajectory design for a UAV which provided ISAC service to jointly optimize the downlink communication rate and localization accuracy. Besides, the freshness of sensed data collected by the UAV in an ISAC system was quantified by the peak age of information (PAoI) in \cite{KZhang2021}, and a joint UAV trajectory, target sensing scheduling, and resource allocation optimization problem was formulated to guarantee the PAoI of the system. Several recent works have considered the multi-UAV ISAC system \cite{KMeng20233,pan2023cooperative,GC2024}. For example, \cite{KMeng20233} proposed a cooperative time-division based sensing and communication scheme and optimized the sensing task allocation. The authors in \cite{pan2023cooperative} proposed an orthogonal frequency division multiple access (OFDMA) UAV-enabled ISAC system and designed a joint trajectory planning and resource allocation scheme. \cite{GC2024} investigated the problem of maximizing the average sum rate of UAVs in the scenario where multiple UAVs communicate with ground base stations and utilize echo signals for detecting functions. However, it is rare to see the mechanism design and applicable technique utilization for emergency UAV-enabled ISAC system, which highly emphasizes the reliability, effectiveness, and efficiency. 

On the other hand, due to the scarcity of spectrum resources, traditional orthogonal multiple access (OMA) struggles to support massive communication in emergency rescue scenarios. Therefore, exploiting multiple-input multiple-output (MIMO) and multiple access techniques with non-orthogonal spectrum has progressed towards the direction of spatial division multiple access (SDMA) and non-orthogonal multiple access (NOMA) into emergency networks \cite{ZXiao2020,SFu2023}. During the emergency events, however, it is worth noting that the transmitted messages are variant with different priority. For example, weather notifications are often broadcasted to all users, while specific rescue instructions are uniquely required by certain users. Such an information demand diversity cannot be simply achieved via SDMA or NOMA that have a weak transmission flexibility in terms of encoding and decoding mechanism. Rate-splitting multiple access (RSMA) is a promising technology to deal with the above issue with higher design flexibility, which allows transmitter to send a superposition of multiple signals to the receiver by splitting the messages into common and private parts with adjustable rates \cite{YMao2022}. Particularly, the common parts are encoded into \emph{common streams} that are decoded by multiple users, while the private parts are independently encoded into the \emph{private streams} that are decoded by the corresponding users only via successive interference cancellation (SIC). Thus, the common parts can transmit some public information such as weather condition and early warning, while the private parts can transmit the intended information of a certain user. Meanwhile, by adjusting the proportion of commonn streams flexibly, RSMA enables more flexible interference management. Due to this characteristics, RSMA is capable of partially decoding the interference while partially treating the remaining interference as noise, which contrasts with SDMA that fully treats interference as noise and NOMA that fully decodes the interference \cite{GZheng2023}. Several literatures have investigated the potential combination of RSMA and ISAC \cite{LYin2022,CXu2021,YLi2021,TT2023}. Specifically, the authors in \cite{LYin2022} proposed a RSMA-assisted ISAC waveform design by jointly minimizing the Cramér-Rao bound (CRB) of the target detection and maximizing the minimum fairness rate (MFR) amongst UEs. In \cite{CXu2021}, the authors considered the optimal transmit beamforming of the communication and radar signals in a multi-antenna RSMA ISAC system. Furthermore, a single UAV-assisted RSMA ISAC system was investigated in \cite{YLi2021}, where the energy efficiency was maximized by optimizing the latitude and transmit beamforming. In \cite{TT2023}, the authors proposed a RSMA-based communication and radar coexistence (CRC) system which significantly improved spectral efficiency, energy efficiency, and quality of service (QoS) of communication users. However, the aforementioned efforts primarily focused on the RSMA ISAC system with a single transceiver, ignoring the performance gain in terms of cooperative sensing and communication, remaining a gap in research and discussion regarding the coordinated RSMA for ISAC with multiple UAVs.

Inspired by the aforementioned analysis, in this paper, considering the diversity of information, a coordinated RSMA ISAC (CoRSMA-ISAC) in emergency UAV system is studied to achieve simultaneous cooperative data transmission with communication survivors (CSs) and target detection for a trapped survivor (TS). Our goal is to improve the communication performance while guaranteeing the sensing requirement by exploiting the transmission flexibility of RSMA and cooperation gain via multiple UAVs. Towards this end, a joint UAV-CS association, UAV deployment, and beamforming optimization problem is formulated to maximize the WSR of the system, subject to the sensing requirement at the TS. The main contributions of this paper are summarized as follows:
\begin{itemize}
	\item First, we propose a CoRSMA-ISAC framework in emergency UAV system where multiple UAVs aim to provide ISAC services for emergency rescue. We first analyze the expression of communication and sensing signals, and derive the corresponding performance metrics, i.e., received sensing signal-to-noise ratio (SNR) and communication signal-to-interference-plus-noise ratio (SINR). To maximize the WSR while satisfying the sensing requirement of TS, a joint deployment and beamforming optimization problem is formulated with the sensing SNR constraint for each CS.
	\item Second, in order to solve the formulated non-convex optimization problem, we decompose the problem into three subproblems, i.e., UAV-CS association, UAV deployment optimization, and beamforming and rate allocation, respectively. Specifically, we first use the K-Means to determine the optimal UAV-CS association. For the UAV deployment, we propose an efficient algorithm utilizing the successive convex approximation (SCA) technique. Moreover, to address the beamforming and rate allocation subproblem, we equivalently transform it into a semi-definite programming (SDP) problem, which is solved by semi-definite relaxation (SDR) and SCA techniques.
	\item Finally, numerical results demonstrate that the proposed design is able to achieve better WSR compared with SDMA, NOMA, and OMA for CSs while guaranteeing the sensing requirement of TS. Furthermore, the common message transmission via different UAVs with coordinated RSMA also provides extra sensing performance gain. Especially, when the total usable transmit power is insufficient or sensing requirements are high, the proposed CoRSMA-ISAC is still able to ensure the C\&S performance by increasing the ratio of common rate.
\end{itemize}

The rest of this paper is organized as follows. Section \ref{SystemModel} presents the system model for the proposed CoRSMA-ISAC in emergency UAV system and optimization problem formulation. Section \ref{ProposedSolution} proposes the algorithm for WSR maximization in CoRSMA-ISAC. Simulation results are presented in Section \ref{SimulationResults}. Moreover, Section \ref{Conclusion} provides the concluding remarks.

\emph{Notations:} In this paper, scalars are denoted by italic letters. Vectors and matrices are denoted by boldface lower and uppercase letters, respectively. For a vector $\mathbf{a}$, its Euclidean norm is denoted as $\|\mathbf{a}\|$. For a matrix $\mathbf{M}$, ${\rm rank}(\mathbf{M})$, $\operatorname{tr}\left(\mathbf{M}\right)$, $\mathbf{M}^T$, $\mathbf{M}^{H}$, $\|\mathbf{M}\|_{\mathrm{F}}$, and $\left[\mathbf{M}\right]_{p,q}$ denote its rank, trace, transpose, conjugate transpose, Frobenius norms, and the element in the $p$-th row and $q$-th column, respectively. Besides, $\mathbf{M}\succeq \mathbf{0}$ represents that $\mathbf{M}$ is a semi-positive definite matrix.  For a complex scalar $a$, its conjugate and magnitude are denoted by $a^{*}$ and $|a|$, respectively. $\mathbb{E}[\mathbf{x}]$ is the expectation of $\mathbf{x}$. For the sake of clarity, a main notation list is included in Table \ref{notation}.
 \begin{table}[h]\label{notation}
 	\centering
	\caption{List of Main Notations}
	\begin{center}
		\begin{tabular}{m{0.3\columnwidth}<{\centering}m{0.6\columnwidth}}
			\toprule[2pt]  
			 \textbf{Notation}    & \textbf{Definition}\\
			\hline $\mathcal{U}$, $\mathcal{K}$    & The set of all the ISAC UAVs and all the CSs.\\
			\hline $\mathcal{O}$, $\mathcal{Q}$ & The set of ISAC UAV position and CS position.\\
			\hline $\mathbf{o}_u$, $\mathbf{q}_k$ & The horizontal coordinate of UAV $u$ and CS $k$.\\
			\hline $ \mathcal{J}_u$ & The CS cluster in which each CS is associated with UAV $u$\\
			\hline $\mathbf{p}_{u,c}$, $ \mathbf{p}_{u,k,p}$, $ \mathbf{p}_{u,r}$ & The beamforming vector for transmitting common message, private message to CS $k$ and sensing waveform at UAV $u$.\\
			\hline $\mathbf{h}_{u,k}$ & The channel corfficient vector  beamforming vector between the UAV $u$ and CS $k$.\\
			\hline $r\left(\mathbf{o}_u,\mathbf{q}_k\right)$, $r\left(\mathbf{o}_u,\mathbf{q}_0\right)$, $r\left(\mathbf{o}_0,\mathbf{q}_0\right)$ & The distance from UAV $u$ to CS $k$, from UAV $u$ to TS, and from receive UAV to TS.\\
			\hline $\mathbf{P}_u$  & The beamforming matrix of UAV $u$, $\mathbf{P}_u=\left[\mathbf{p}_{u,c}, \mathbf{p}_{ u, 1,p}, \ldots, \mathbf{p}_{ u, K,p}, \mathbf{p}_{u,r}\right] \in \mathbb{C}^{N_t \times(K+2)}$ \\
			\hline $\mathbf{P}_{u,c}$, $\mathbf{P}_{u, k,p}$, $\mathbf{H}_{u, k}$ & Semi positive definite matrix, where $\mathbf{P}_{u,c}=\mathbf{p}_{u,c}\mathbf{p}_{u,c}^H$, $\mathbf{P}_{u, k,p}=\mathbf{p}_{u, k,p}\mathbf{p}_{u, k,p}^H$,$\mathbf{H}_{u, k}=\mathbf{h}_{u, k}\mathbf{h}_{u, k}^H$ \\
			\hline $R^p_k$ & The achievable private rate of CS $k$. \\
			\hline $C_k$  & The common rate allocation variable, represents the allocated common rate to CS $k$\\
			\bottomrule[2pt]
		\end{tabular} 
	\end{center}
	\label{table:notations}
\end{table}

\section{System Model and Problem Formulation}\label{SystemModel}
\subsection{System Model}\label{system_model}
\begin{figure}[t]
	\center
	\includegraphics[width=3.3in]{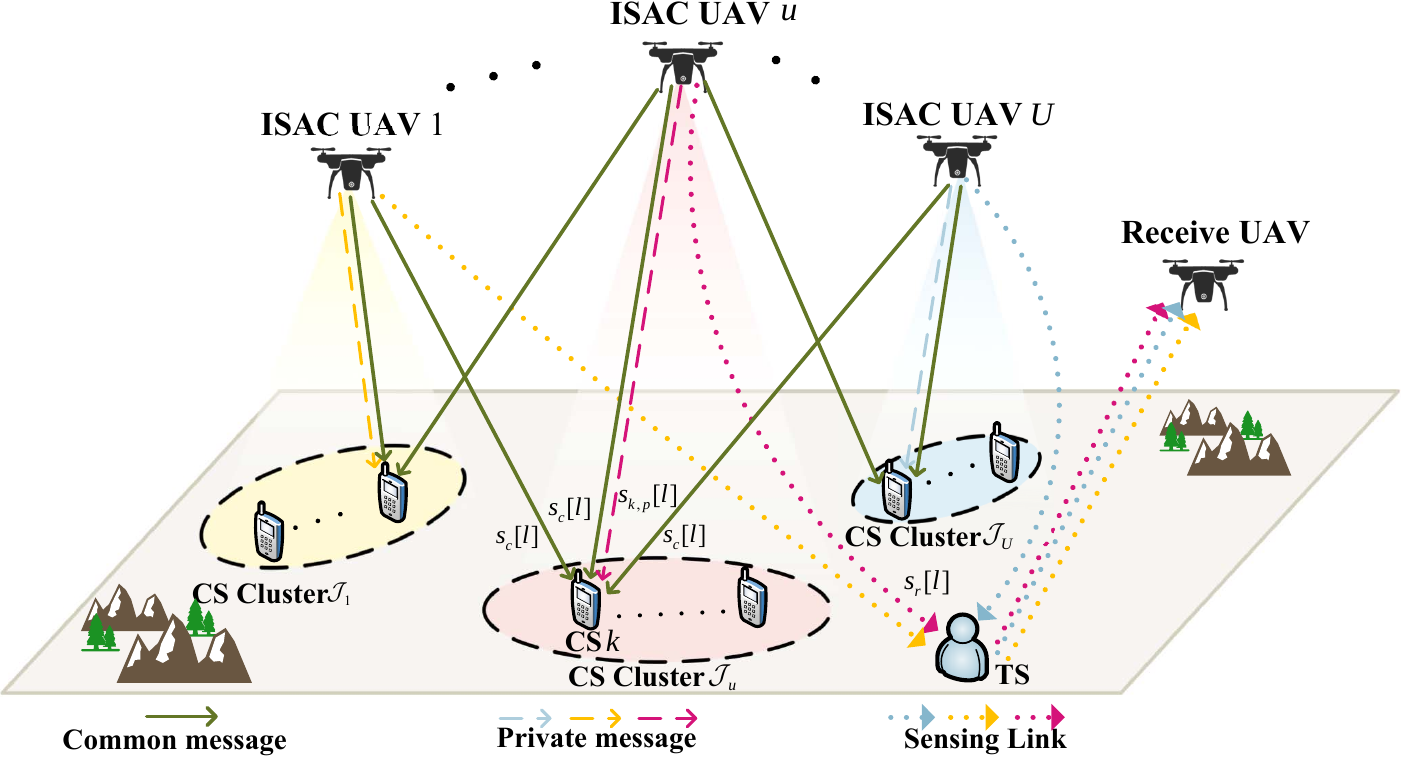}\\
	\caption{Coordinated RSMA for integrated sensing and communication (CoRSMA-ISAC) in emergency UAV system.}\label{Fig_1}
\end{figure}
As illustrated in Fig. \ref{Fig_1}, we consider a CoRSMA-ISAC framework in emergency UAV system to accomplish a post-disaster rescue mission, which consists of $U$ ISAC UAVs and a dedicated radar receive UAV \cite{Li2023}. The system provides downlink transmissions for $K$ CSs, while cooperatively detecting a single TS whose mobile device is out of work\footnote{We assume that all the clutters can be mitigated by using the existing clutter suppression
techniques \cite{XChen2020}, such that we do not consider the influence of the clutters.}. The sets of ISAC UAVs and CSs are denoted by $\mathcal{U}=\{1,\cdots,U\}$ and $\mathcal{K}\in\{1,\cdots,K\}$, respectively. We assume that each ISAC UAV is equipped with $N_t$ transmit antennas and sensing receive UAV is equipped with $N_r$ receive antennas, while each CS is equipped with single antenna. We further consider a three-dimensional (3D) Cartesian coordinate system, where the horizontal coordinate of CS $k, k\in\mathcal{K}$ and TS are fixed at $\mathbf{q}_k = (x_{k,q}, y_{k,q})$ and $\mathbf{q}_0 = (x_{0,q}, y_{0,q})$, respectively. The ISAC UAV $u$, $u\in\mathcal{U}$ and receive UAV are assumed to fly at fixed altitude $H_u$ and $H_0$, respectively \cite{Wu2018jointtraj}. Let us further denote $\mathbf{o}_u = (x_{u,o}, y_{u,o})$ and $\mathbf{o}_0 = (x_{0,o}, y_{0,o})$, as their horizontal location, respectively. The sets of horizontal coordinates of ISAC UAV and CS can be denoted as $\mathcal{O}=\{\mathbf{o}_1 ,\cdots,\mathbf{o}_U\}$ and $\mathcal{Q}=\{\mathbf{q}_1 ,\cdots,\mathbf{q}_K\}$, respectively. 

In order to achieve dual purposes of communication and sensing, coordinated RSMA is adopted. Specifically, the communication messages for CS $k$ are split into two parts by each UAV, i.e., a \emph{common part} that carries the public information needed by all CSs, such as weather conditions, rescue situations, etc.,  and a \emph{private part} that is exclusively desired by CS $k$. We denote $ \mathcal{J}_u$ as the CS cluster containing all the CSs associated with UAV $u$, in which if CS $k \in \mathcal{J}_u$, CS $k$ will only receive the private message from UAV $u$, i.e., $\mathcal{J}_u\cap \mathcal{J}_j=\emptyset, \forall u,j \in \mathcal{U}, u\neq j$. Additionally, we assume that the ISAC signal transmission over a particular block consisting of $L$ symbols, and the set of symbols is denoted as $\mathcal{L}=\{1, \cdots, L\}$. The messages for CSs are jointly processed at the central controller and the optimized transmit signals are sent to the corresponding UAVs \cite{YM2018}. For the $l$-th symbol, the common parts of all $K$ CSs are jointly encoded into a single common message $s_c[l]$ that is transmitted by all ISAC UAVs, and the private parts are encoded separately into $K$ private messages $s_{k,p}[l],\forall k\in{\mathcal{K}}$. Furthermore, to fully exploit the spatial degree of freedom (DoF) of the transmit antennas, a dedicated sensing waveform $s_{r}[l]$ is introduced to enhance the sensing performance. In this system model, $s_c[l]$, $s_{k,p}[l]$, and $s_{r}[l]$ are all assumed to have zero mean and unit power, i.e., $\mathbb{E}\left\{\left|s_c[l]\right|^2\right\}=\mathbb{E}\left\{\left|s_{k,p}[l]\right|^2\right\}=\mathbb{E}\left\{\left|s_{r}[l]\right|^2\right\}=1$, and they are independent from each other \cite{ZB2023}. Then the baseband transmitted signal from ISAC UAV  $u$ can be expressed as 
\begin{equation}
	\mathbf{x}_u[l]=\mathbf{p}_{u,c} s_c[l]+\sum_{k \in \mathcal{J}_u} \mathbf{p}_{u,k,p} s_{k,p}[l]+\mathbf{p}_{u,r} s_{r}[l],
\end{equation}
where $\mathbf{p}_{u,c}\in \mathbb{C}^{N_t \times 1}$, $ \mathbf{p}_{u,k,p} \in \mathbb{C}^{N_t \times 1}$, and $ \mathbf{p}_{u,r}\in \mathbb{C}^{N_t \times 1} $ are the corresponding beamforming vector, respectively.

\subsection{Communication Model}

Let $\mathbf{h}_{u,k}=\left[h_{u, k}^1,\dots,h_{u, k}^n,\dots,h_{u, k}^{N_t}\right]^T \in \mathbb{C}^{N_t \times 1}$ denote the channel coefficient vector between the UAV $u$ and CS $k$. Specifically, $h_{u, k}^n, n\in\{1,\cdots,N_t\},$ is the cofficient between the $n$-th antenna of the UAV $u$ and CS $k$, which can be modeled as \cite{YZeng2019}
\begin{align}
h_{u, k}^n=\sqrt{\varepsilon_{u,k}}\tilde{h}_{u, k}^n,\forall n\in\{1,\cdots,N_t\},
\end{align}
where $\tilde{h}_{u, k}^n$ denotes the small-scale fading due to the multipath propagation with $\mathbb{E}[|\tilde{h}_{u, k}^n|^2]=$ 1 \cite{pan2023cooperative}. $\varepsilon_{u,k}$ is the larger-scale channel power, which is given by
\begin{align}
\varepsilon_{u,k}=\frac{\varepsilon_0}{r^2\left(\mathbf{o}_u,\mathbf{q}_k\right)},
\end{align}
where $r\left(\mathbf{o}_u,\mathbf{q}_k\right)=\sqrt{\|\mathbf{o}_u-\mathbf{q}_k\|^2 +{H_u}^2}$ represents the distance between UAV $u$ and CS $k$. $\varepsilon_0\triangleq\frac{G_\mathrm{T} G_\mathrm{C}\lambda^2}{(4\pi)^2}$ denotes the channel power at the reference distance of 1 meter  \cite{XJing2022}, where $G_{\mathrm{T}}$ and $G_\mathrm{C}$ are the transmit and receive antenna gain of the CS, respectively. $\lambda$ is the wavelength. Therefore, the received signal at CS $k\in\mathcal{K}$ is given by 
\begin{align}
	&y_k[l]=\underbrace{\sum_{u \in \mathcal{U}} \mathbf{h}_{u, k}^H \mathbf{p}_{u,c} s_c[l]}_{\text{Desired common signal}}+\underbrace{\mathbf{h}_{u, k}^H \mathbf{p}_{u, k,p} s_{k,p}[l]}_{\text{Desired private signal}}\nonumber\\
	&+\underbrace{\sum_{j \in \mathcal{J}_u \backslash \{k\}} \mathbf{h}_{u, j}^H \mathbf{p}_{u, j,p} s_{j,p}[l]+\sum_{i \in \mathcal{U} \backslash \{u\}}\sum_{j \in \mathcal{J}_i} \mathbf{h}_{i, j}^H \mathbf{p}_{i, j,p} s_{j,p}[l]}_{\text{Multi-user intra and intercell Interference}}\nonumber\\
	&+\underbrace{\mathbf{h}_{u, k}^H \mathbf{p}_{u,r}s_{r}[l]}_{\text{Sensing interference}}+\underbrace{n_k[l]}_{\text{Noise}},
\end{align}
where $n_k[l]$ represents the additive white Gaussian noise (AWGN) with zero mean and variance $\sigma^2$. Since the dedicated radar waveform $s_{r}[l]$ can be a prior known to CS $k$, it can be removed from the received signal \cite{Li2023}. Following the decoding order of RSMA, each CS first decodes the common message $s_c[l]$ by treating other streams as noise \cite{YMao2022}. Specifically, the received SINR of decoding the common message $s_c[l]$ at CS $k$ can be expressed by
\begin{equation}\label{rc}
		\gamma_k^c=\frac{\left|\sum_{u \in \mathcal{U}}\mathbf{h}_{u, k}^H \mathbf{p}_{u,c}\right|^2}{ \sum_{u \in \mathcal{U}} \sum_{j \in \mathcal{J}_u}  \left|\mathbf{h}_{u, k}^H \mathbf{p}_{u, j,p}\right|^2+\sigma^2},
\end{equation}
and the corresponding achievable rate is given by
\begin{equation}
	R_k^c=B\log _2\left(1+\gamma_k^c\right),
\end{equation}
where $B$ is the channel bandwidth. Besides, to guarantee that all CSs are capable of decoding the common stream, the common rate is defined as
\begin{equation}
	R^c=\min _k\left\{R_k^c \mid k \in \mathcal{K}\right\}.
\end{equation}
Since $s_c[l]$ contains the common messages for $K$ CSs, $R^c$ is accordingly shared by $ K $ users. We denote the variable $C_k$ as the allocated common rate to CS $k$. Then we have
\begin{equation}
	\sum_{k \in \mathcal{K}} C_k\leq R^c.
\end{equation}
 
After successfully decoding and removing the common message, the SINR for CS $k$ to decode its own private message can be written as \eqref{rp} on the top of next page.
\begin{figure*}
	\begin{equation}\label{rp}
		\gamma_k^p=\frac{\left|\mathbf{h}_{u, k}^H \mathbf{p}_{u, k,p}\right|^2}{\sum_{j \in \mathcal{J}_u \backslash \{k\}} \left|\mathbf{h}_{u, k}^H \mathbf{p}_{u, j,p}\right|^2+ \sum_{i \in \mathcal{U} \backslash \{u\}} \sum_{j \in \mathcal{J}_i} \left|\mathbf{h}_{i, k}^H \mathbf{p}_{ i, j,p}\right|^2+\sigma^2}.
	\end{equation}
\hrulefill
\end{figure*}
And the corresponding achievable private rate is
\begin{equation}
	R_k^p=B\log _2\left(1+\gamma_k^p\right).
\end{equation}
Then the total achievable rate of CS $k$ is equal to the summation of the allocated common rate and the private rate, which can be expressed as
\begin{equation}\label{Rp}
	R^{tot}_k=C_k+R_k^p,\forall k \in \mathcal{K}.
\end{equation}
Thus, the WSR of the system is written by
\begin{equation}
	 R^{w} = \sum_{k\in\mathcal{K}}\mu_kR^{tot}_k = \sum_{k \in \mathcal{K}} \mu_k\left(C_k+R_k^p\right),
\end{equation}
where $\mu_k\in[0,1]$ denotes the rate weight allocated to CS $k$ with $\sum_{k=1}^{K}\mu_k=1$. 

\subsection{Sensing Model}

For the sensing model, we assume that TS is a point-like target. Therefore, the channel matrix between the ISAC UAV $u$ and the receive UAV through the TS can be defined as
\begin{equation}\label{Gu}
	\mathbf{G}_u=\beta_u  \mathbf{b}\left(\mathbf{o}_0,\mathbf{q}_0\right) \mathbf{a}^H\left(\mathbf{o}_u,\mathbf{q}_0\right),
\end{equation}
where $\beta_u$ represents the total sensing channel power gain, which is expressed as
\begin{equation}
	\beta_u=\sqrt{\frac{\beta _0}{r^2\left(\mathbf{o}_u,\mathbf{q}_0\right)r^2\left(\mathbf{o}_0,\mathbf{q}_0\right)}},
\end{equation} 
where $r\left(\mathbf{o}_u,\mathbf{q}_0\right)=\sqrt{\left\|\mathbf{o}_u-\mathbf{q}_0\right\|^2 +{H_u}^2}$ denotes the distance between UAV $u$ and TS, and $r\left(\mathbf{o}_0,\mathbf{q}_0\right)=\sqrt{\left\|\mathbf{o}_0-\mathbf{q}_0\right\|^2 +{H_0}^2}$ represents the distance between the receive UAV and TS. $\beta_{0}$ denotes the channel power at the reference distance of 1 meter, which is expressed as
\begin{equation}
	\beta_{0}=\frac{G_{\mathrm{T}} G_{\mathrm{S}}\sigma_{\mathrm{rcs}}\lambda^{2}}{\left(4\pi\right)^{3}},
\end{equation} 
where $G_{\mathrm{S}}$ is the receive antenna gain of the target and $\sigma_{\mathrm{rcs}}$ is the radar cross section (RCS). $\mathbf{a}\left(\mathbf{o}_u,\mathbf{q}_0\right)\in\mathbb{C}^{N_t\times 1}$ and $  \mathbf{b}\left(\mathbf{o}_0,\mathbf{q}_0\right)\in\mathbb{C}^{N_r\times 1}$ in \eqref{Gu} denote the transmit and receive steering vectors of the transmit and receive antennas, which are assumed to be a uniform linear array (ULA) with half-wavelength antenna spacing. Therefore, $\mathbf{a}\left(\mathbf{o}_u,\mathbf{q}_0\right)$ and $\mathbf{b}\left(\mathbf{o}_0,\mathbf{q}_0\right)$ can be respectively given by 
\begin{align}
	&\mathbf{a}\left(\mathbf{o}_u,\mathbf{q}_0\right)\nonumber\\
	&=\left[1, e^{j \pi \cos \theta \left(\mathbf{o}_u,\mathbf{q}_0\right)}, \cdots, e^{j  \pi\left(N_t-1\right) \cos \theta \left(\mathbf{o}_u,\mathbf{q}_0\right)}\right]^T,
\end{align}
\begin{align}
	&\mathbf{b}\left(\mathbf{o}_0,\mathbf{q}_0\right)\nonumber\\
	&=\left[1, e^{j \pi  \cos \phi \left(\mathbf{o}_0,\mathbf{q}_0\right)}, \cdots, e^{j  \pi\left(N_r-1\right) \cos \phi \left(\mathbf{o}_0,\mathbf{q}_0\right)}\right]^T,
\end{align}
where $\theta\left(\mathbf{o}_u,\mathbf{q}_0\right)=\arccos \frac{H_u}{r\left(\mathbf{o}_u,\mathbf{q}_0\right)}$ and $\phi\left(\mathbf{o}_0,\mathbf{q}_0\right)=\arccos \frac{H_0}{r\left(\mathbf{o}_0,\mathbf{q}_0\right)}$ denote the transmit and receive azimuth angles, respectively.

Therefore, the received signal reflected from TS at the receive UAV is written as
\begin{equation}
	\resizebox{1\hsize}{!}{$\begin{aligned}
	 {\mathbf{y}}_{r}[l]=\sum_{u \in \mathcal{U}} \mathbf{G}_u \mathbf{x}_u[l]+\mathbf{n}[l]=\sum_{u \in \mathcal{U}} \beta_u \mathbf{b}\left(\mathbf{o}_0,\mathbf{q}_0\right) \mathbf{a}^H\left(\mathbf{o}_u,\mathbf{q}_0\right) \mathbf{x}_u[l]+\mathbf{n}[l],
\end{aligned}$}
\end{equation}
where $\mathbf{n}[l] \in \mathbb{C}^{N_r \times 1}$ is the AWGN vector following $\mathbf{n}[l] \sim \mathcal{C N}\left(0, \sigma^2 \mathbf{I}\right)$ with $\sigma^2$ denoting the variance of each entry.

For notational simplicity, let us denote
\begin{align}
	\mathbf{P}_u=&\left[\mathbf{p}_{u,c}, \mathbf{p}_{ u, 1,p}, \ldots, \mathbf{p}_{ u, K,p}, \mathbf{p}_{u,r}\right] \in \mathbb{C}^{N_t \times(K+2)},\\
	\mathbf{s}[l]=&[s_c[l],s_{1,p}[l],\cdots,s_{K,p}[l],s_{r}[l]]^T \in\mathbb{C}^{(K+2)\times 1}, l\in\mathcal{L},\\
	\mathbf{S}=&\left[\mathbf{s}[1],\cdots,\mathbf{s}[L]\right] \in \mathbb{C}^{(K+2) \times L}.
\end{align}
Then, the transmit signal with $L$ symbols from UAV $u$ can be expressed as
\begin{equation}
	\mathbf{X}_{u}=\mathbf{P}_u \mathbf{S}\in \mathbb{C}^{N_t \times L}.
\end{equation}
Thus, the signal at the receive UAV can be writen as 
\begin{align}
	\mathbf{Y}_r=&\sum_{u \in \mathcal{U}} \mathbf{G}_u \mathbf{X}_u + \mathbf{N}\\\nonumber
	=&\underbrace{\sum_{u \in \mathcal{U}} \underbrace{\beta_u \mathbf{b}\left(\mathbf{o}_0,\mathbf{q}_0\right) \mathbf{a}^H\left(\mathbf{o}_u,\mathbf{q}_0\right) \mathbf{P}_u}_{\triangleq\bar{\mathbf{G}}_u}}_{\triangleq\mathbf{G}} \mathbf{S}+\mathbf{N},
\end{align}
where $\mathbf{N}=\left[\mathbf{n}[1],\ldots,\mathbf{n}[L]\right]$.
And the received sensing SNR can be written as
\begin{align}\label{rs}
	\mathrm{\gamma^s}&=\frac{\mathbb{E}\left[\|\mathbf{G} \mathbf{S}\|_{\mathrm{F}}^2\right]}{\mathbb{E}\left[\|\mathbf{N}\|_{\mathrm{F}}^2\right]}=\frac{\left|\beta_0\right|}{r^2\left(\mathbf{o}_0,\mathbf{q}_0\right) \sigma^2}\sum_{u \in \mathcal{U}}\frac{\|\mathbf{a}^H\left(\mathbf{o}_u,\mathbf{q}_0\right) \mathbf{P}_u\|^2} {r^2\left(\mathbf{o}_u,\mathbf{q}_0\right) } .
\end{align}
The derivation of \eqref{rs} is provided in Appendix \ref{proof-rs}.

\subsection {Problem Formulation}
We aim to maximize the WSR of the system by joint optimizing the UAV-CS association clusters $\left\{\mathcal{J}_u\right\}_{u=1}^{U} $, the UAV deployment $\mathcal{O}$, the transmit beamforming matrix $\left\{\mathbf{P}_u\right\}_{u=1}^{U}$ and common rate allocation variables $\left\{C_k\right\}_{k=1}^{K}$, subject to the sensing performance requirements, transmission QoS, and power limits, which can be formulated as
\begin{small}
\begin{subequations}\label{P1}
	\begin{align}
		(\text{P0}):&\max _{\left\{\mathcal{J}_u\right\}_{u=1}^{U}, \mathcal{O},\atop \left\{\mathbf{P}_u\right\}_{u=1}^{U},\left\{C_k\right\}_{k=1}^{K} }\quad \sum_{k \in \mathcal{K}} \mu_k\left(C_k+R_k^p\right)\\
		\mbox{s.t.}
		\label{st1}&\quad\mathcal{J}_u\cap \mathcal{J}_j=\emptyset, \forall u,j \in \mathcal{U}, u\neq j,\\
		\label{st2}&\quad\sum_{u \in \mathcal{U}}\left|\mathcal{J}_u\right| =K ,\forall \mathcal{J}_u\neq \emptyset,\forall  u \in \mathcal{U}, \\
		\label{st3}&\quad\sum_{k \in \mathcal{K}} C_{k} \leq \min _{i}\left\{R_{i}^{c} \mid i \in \mathcal{K}\right\},\\
		\label{st4}&\quad C_k\geq\text{0},\forall k\in\mathcal{K},\\
		\label{st5}&\quad \text{tr}\left(\mathbf{P}_u{\mathbf{P}_u}^H\right) \leq P_{max},\forall u \in \mathcal{U},\\
		\label{st6}&\quad C_k+R_k^p\geq R_k^{th}\text{,}\forall k\in\mathcal{K},\\
		\label{st7}&\quad \mathrm{\gamma^s}\geq\bar\gamma,
	\end{align}
\end{subequations}
\end{small}where \eqref{st1} guarantees that each CS is only associated with one UAV. \eqref{st2} ensures all CSs are associated with the UAVs, where $\left|\mathcal{J}_u\right|$ is the number of CSs served by UAV $u$. \eqref{st3} ensures the common stream can be decoded by all users. \eqref{st4} implements a feasible partition of the common stream. \eqref{st5} ensures that the transmit power of each UAV meets the power budget $P_{max}$. \eqref{st6} indicates the communication rate requirement of each CS, where $R_k^{th}$ is the threshold for the total achievable rate of CS $k$. \eqref{st7} indicates the sensing performance requirement, where $\bar{\gamma}$ is the predefined sensing SNR threshold. 
\section{Proposed Algorithm for Weighted Sum Rate Maximization in CoRSMA-ISAC}\label{ProposedSolution}

The optimization problem \text{(P0)} is a non-convex optimization problem, which is NP-hard in general and cannot be solved directly by standard convex optimization tools. In this regard, we propose an efficient algorithm to solve the problem \text{(P0)} which first determines the UAV-CS association cluster  $\left\{\mathcal{J}_u\right\}_{u=1}^{U}$ via K-Means, then alternately optimizes the UAV deployment $\mathcal{O}$ and beamforming matrix $\left\{\mathbf{P}_u\right\}_{u=1}^{U}$ and common rate allocation variables $\left\{C_k\right\}_{k=1}^{K}$ by fixing one and solve remaining one. The framework of the proposed algorithm is illustrated in Fig. \ref{Scheme}.
\begin{figure}[h]
	\center
	\includegraphics[width=3.3in]{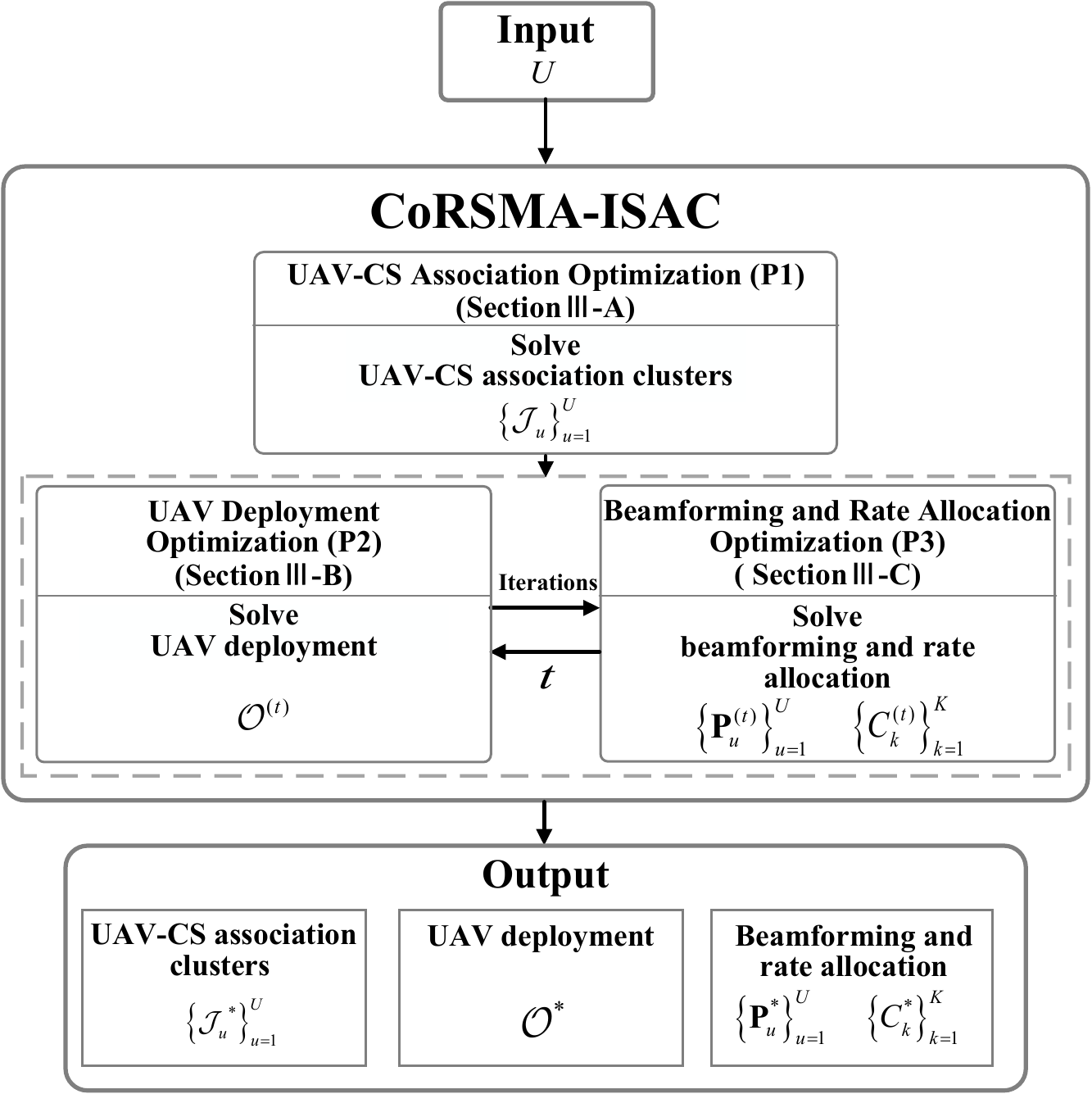}\\
	\caption{Framework of the proposed algorithm to the joint UAV-CS association, UAV deployment, and beamforming and rate allocation optimization problem for CoRSMA-ISAC in emergency UAV system.}\label{Scheme}
\end{figure}

\subsection{UAV-CS Association Optimization}\label{UAV-CS-Association-Optimization}
We first update the UAV-CS association clusters $\left\{\mathcal{J}_u\right\}_{u=1}^{U} $ with the fixed transmit beamforming matrix $\left\{\mathbf{P}_u\right\}_{u=1}^{U}$ and common rate allocation variable $\left\{C_k\right\}_{k=1}^{K}$. Then, the optimization problem \text{(P0)} reformulated as

\begin{subequations}\label{P2}
	\begin{align}
		(\text{P1}):\max _{\left\{\mathcal{J}_u\right\}_{u=1}^{U}  }\quad &\sum_{k \in \mathcal{K}} \mu_k R_k^p\\
		\mbox{s.t.}\quad
		\label{2st1}&\mathcal{J}_u\cap \mathcal{J}_j=\emptyset, \forall u,j \in \mathcal{U}, u\neq j,\\
		\label{2st2}&\sum_{u \in \mathcal{U}}\left|\mathcal{J}_u\right| =K ,\forall \mathcal{J}_u\neq \emptyset,\forall  u \in \mathcal{U},\\
		\label{2st3}&C_k+ R_k^p\geq R_k^{th}\text{,}\forall k\in\mathcal{K},\\
		\label{2st5}&\mathrm{\gamma^s}\geq\bar\gamma.
	\end{align}
\end{subequations}
The problem \text{(P1)} is an integer programming problem, which can be solved by K-Means algorithm. First, we determine the number of clusters based on the number of UAVs $U$ and cluster CSs based on CSs position. Specifically, we denote $\mathcal{J}_u^{(\nu)}$ as $u$-th cluster at $\nu$-th step, and treat the cluster centroids as $\mathbf{o}_u^{(\nu)}, u \in \mathcal{U}$, with $\mathcal{O}^{(\nu)}$ denoting the position set of cluster centroids. 

At the $\nu$-th step, we first calculate the distances between CSs and cluster centroids, and assign CS $k$ into the cluster corresponding
to the nearest cluster centroid $\mathcal{J}_u^{(\nu)}$.
\begin{equation}\label{kmeans-J}
	k \in \mathcal{J}_u^{(\nu)} \Leftrightarrow	r\left(\mathbf{o}_u^{(\nu)},\mathbf{q}_k\right)= \min_{i\in U}\left\{r\left(\mathbf{o}_i^{(\nu)},\mathbf{q}_k\right)\right\}.
\end{equation}
    
Next, update the $u$-th cluster centroid by averaging the geographic positions of clustered CSs, which is denoted by
\begin{equation}\label{kmeans-o}
	\mathbf{o}_{u}^{(\nu+1)}=\frac{\sum_{k\in \mathcal{J}_{u}}\mathbf{q}_k}{|\mathcal{J}_{u}|}, \forall u \in \mathcal{U}.
\end{equation}
Then update clusters $\mathcal{J}_u^{(\nu+1)}$ based on \eqref{kmeans-J} until the cluster centroids no longer change. Then we get the clusters $\left\{\mathcal{J}_u^{(\nu+1)}\right\}_{u=1}^{U}$ as the UAV-CS association clusters and the position set of cluster centroids $\mathcal{O}^{(\nu+1)}$, which can be used as initialization of the UAV deployment optimization.
The process of using the K-Means method to solve problem \text{(P1)} is summarized in Algorithm \ref{kmeans}.

\begin{algorithm}[tt]
	\renewcommand{\algorithmicrequire}{\textbf{Inputs:}}
	\renewcommand{\algorithmicensure}{\textbf{Outputs:}}
	\caption{UAV-CS Association for Problem \text{(P1)} By K-Means}
	\label{kmeans}
	\begin{algorithmic}[1]
		\REQUIRE{ $\mathcal{Q}$. Fixed UAV height $z_{u,o}=z_0,\forall u\in\mathcal{U}$. Set iteration number $\nu$ }.
		\ENSURE  the  UAV-CS association clusters $\left\{\mathcal{J}_u\right\}_{u=1}^{U}$. \\
		\STATE {\bf{begin:}}
		\STATE $\nu=0$
		\STATE  Randomly initialize the cluster center $\mathcal{O}^{(\nu)}$.
		\REPEAT
		\STATE Update the clusters $\left\{\mathcal{J}_u^{(\nu+1)}\right\}_{u=1}^{U}$ based on \eqref{kmeans-J}.
		\STATE Update the cluster centroid $\mathcal{O}^{(\nu+1)}$ based on\eqref{kmeans-o}.
		\STATE $\nu=\nu+1$
		\UNTIL{$\mathcal{O}^{(\nu)}=\mathcal{O}^{(\nu+1)} $ }
		\STATE Set $\mathcal{J}_{u}=\mathcal{J}_{u}^{(\nu)},\forall u\in\mathcal{U}$
	\end{algorithmic}
\end{algorithm}

\subsection{UAV Deployment Optimization}\label{UAVDeploymentOptimization}
The UAV deployment $\mathcal{O}$ is then optimized with the fixed transmit beamforming vectors $\left\{\mathbf{P}_u\right\}_{u=1}^{U}$, common rate allocation variable $\left\{C_k\right\}_{k=1}^{K}$, and updated UAV-CS association clusters $\left\{\mathcal{J}_u\right\}_{u=1}^{U}$ and is initialized to $\mathcal{O}^0$ based on the position set of cluster centroids obtained in problem \text{(P1)}. Specifically, by introducing a slack variable $f_{u,k}$, the optimization problem \text{(P0)} can be reformulated as
\begin{subequations}\label{P30}
	\begin{align}
	(\text{P2}):\max _{ \mathcal{O}}\quad &\sum_{k \in \mathcal{K}} \mu_k f_{u,k}&\\
		\mbox{s.t.}\quad
			\label{3st2}& f_{u,k}\leq R_k^p,\forall k \in \mathcal{K},\\
			\label{3st3}&C_k+R_k^p\geq R_k^{th}\text{,}\forall k\in\mathcal{K},\\
			\label{3st4}&\mathrm{\gamma^s}\geq\bar\gamma.
	\end{align}
\end{subequations}

Problem \text{(P2)} is also a non-convex problem due to the constraints \eqref{3st2}, \eqref{3st3} and \eqref{3st4}. Specifically, to deal with the non-convex constraint \eqref{3st3}, we introduce slack variables $\{\hat{r}_{i,k}\}$ that follow
\begin{equation}\label{r-ik}
	\hat{r}_{i,k}\leq r^2\left(\mathbf{o}_i,\mathbf{q}_k\right),\forall k \in \mathcal{K}, k\notin \mathcal{ J }_i .
\end{equation}

\begin{prop}\label{prop1}
 With \eqref{r-ik}, \eqref{3st3} can be reformulated as a convex constraint, which is expressed as
	\begin{align}\label{3st3-o}
		\resizebox{1\hsize}{!}{$\begin{aligned}\left\|\mathbf{o}_u-\mathbf{q}_k\right\|^2\leq\frac{1}{\Psi}\left(\frac{\left|\mathds{1} \mathbf{p}_{ u, k,p}\right|^2}{2^{\frac{R_k^{t h}-C_k}{B}}-1}-\sum_{j \in \mathcal{ J }_u\backslash\{k\}}\left|\mathds{1}  \mathbf{p}_{u, j,p}\right|^2\right)-{H_u}^2,
		\end{aligned}$}
	\end{align}
where
\begin{align}
	\Psi&=\sum_{i \in \mathcal{U} \backslash\{u\}} \sum_{j \in \mathcal{J}_i} \frac{1}{	\hat{r}_{i,k}}\left|\mathds{1} \mathbf{p}_{i, j,p}\right|^2+\sigma^2\varepsilon_0^{-1},
\end{align}
 with $\mathds{1} \in \mathbb{R}^{1\times N_t}$ being an all-ones vector.
\end{prop}
\begin{proof}
	Please refer to Appendix \ref{proof1}.
\end{proof}

Next, we deal with the non-convex constraint \eqref{3st2}. Similarly, with \eqref{r-ik}, constraint \eqref{3st2} can be reformulated as
\begin{equation}\label{st-SLACK-Rp}
	f_{u,k}\leq \tilde{R}^p_k,\forall k \in \mathcal{K},
\end{equation}
where $\tilde{R}^p_k$ is expressed in \eqref{slack-Rp} at the top of this page.  
\newcounter{mytempeqncnt1}
\setcounter{mytempeqncnt1}{\value{equation}}
\begin{figure*}
	\begin{equation}\label{slack-Rp}
		\tilde{R}_k^p=B\log_2\left(1+\frac{\left|\mathds{1} \mathbf{p}_{u, k,p}\right|^2}{\sum_{j \in \mathcal{J}_u\backslash\{k\}}\left|\mathds{1}\mathbf{p}_{u, j,p}\right|^2+\sum_{i \in U \backslash\{u\}} \sum_{j \in \mathcal{J}_i} \frac{r^2\left(\mathbf{o}_u,\mathbf{q}_k\right)}{\hat{r}_{i,k}}\left|\mathds{1}\mathbf{p}_{ i, j,p}\right|^2+\sigma^2\varepsilon_0^{-1} r^2\left(\mathbf{o}_u,\mathbf{q}_k\right)}\right),
	\end{equation}
\setcounter{mytempeqncnt1}{\value{equation}}
\hrulefill
\end{figure*}We adopt the SCA to deal with it. Specifically, by using the fact that the first-order Taylor expansion of the convex differentiable functions serves as a global lower bound \cite{jaafar2020}, we have 
\begin{equation}\label{3st3-1}
	f_{u,k} \leq A_{u,k}^{(\upsilon)}+{\mathbf{c}_{u,k}^{(\upsilon)}}^H\left(\mathbf{o}_u-\mathbf{o}_u^{(\upsilon)}\right),
\end{equation}
where $\mathbf{o}_u^{(\upsilon)}$ is the local point obtained at the $\upsilon$-th iteration.   $A_{u,k}^{(\upsilon)}$ and $\mathbf{c}_{u,k}^{(\upsilon)}$ are expressed in \eqref{A} and \eqref{B}, respectively, at the top of this page.
	
\newcounter{mytempeqncnt}
\setcounter{mytempeqncnt}{\value{equation}}
\begin{figure*}
	\vspace{-14pt}
	\begin{align}\label{A}
		A_{u,k}^{(\upsilon)}=B\log_2\left(1+\frac{\left|\mathds{1} \mathbf{p}_{u, k,p}\right|^2}{\sum_{j \in \mathcal{J}_u\backslash\{k\}}\left|\mathds{1}\mathbf{p}_{u, j,p}\right|^2+\sum_{i \in U \backslash\{u\}} \sum_{j \in \mathcal{J}_i} \frac{r^2\left(\mathbf{o}_u^{(\upsilon)},\mathbf{q}_k\right)}{\hat{r}_{i,k}}\left|\mathds{1}\mathbf{p}_{i, j,p}\right|^2+\sigma^2\varepsilon_0^{-1} r^2\left(\mathbf{o}_u^{(\upsilon)},\mathbf{q}_k\right)}\right),
	\end{align} 

	\begin{align}\label{B}
		&\mathbf{c}_{u,k}^{(\upsilon)}=\frac{-2B\log _2(e)\left|\mathds{1} \mathbf{p}_{u, k,p}\right|^2 \cdot\left(\sum_{i \in U \backslash\{u\}} \sum_{j \in \mathcal{J}_i} \frac{1}{\hat{r}_{i,k}}\left|\mathds{1} \mathbf{p}_{i, j,p}\right|^2+\sigma^2\varepsilon_0^{-1}\right)\left(\mathbf{o}_u^{(\upsilon)}-\mathbf{q}_k\right)}{C_{u,k}^{(\upsilon)}D_{u,k}^{(\upsilon)}},\\
		\text{where } &C_{u,k}^{(\upsilon)}=\left(\sum_{j \in \mathcal{J}_u \backslash\{k\}}\left|\mathds{1} \mathbf{p}_{u, j,p}\right|^2+\sum_{i \in \mathcal{U} \backslash\{u\}} \sum_{j \in \mathcal{J}_i} \frac{r^2\left(\mathbf{o}_u^{(\upsilon)},\mathbf{q}_k\right)}{\hat{r}_{i,k}}\left|\mathds{1} \mathbf{p}_{i, j,p}\right|^2+\sigma^2\varepsilon_0^{-1} r^2\left(\mathbf{o}_u^{(\upsilon)},\mathbf{q}_k\right)\right),\nonumber\\
		&D_{u,k}^{(\upsilon)}=\left(\sum_{j \in \mathcal{J}_u}\left|\mathds{1} \mathbf{p}_{u, j,p}\right|^2+\sum_{i \in \mathcal{U} \backslash\{u\}} \sum_{j \in \mathcal{J}_i} \frac{r^2\left(\mathbf{o}_u^{(\upsilon)},\mathbf{q}_k\right)}{\hat{r}_{i,k}}\left|\mathds{1} \mathbf{p}_{i, j,p}\right|^2+\sigma^2\varepsilon_0^{-1} r^2\left(\mathbf{o}_u^{(\upsilon)},\mathbf{q}_k\right)\right).\nonumber
	\end{align}
\setcounter{mytempeqncnt}{\value{equation}}
\hrulefill
\end{figure*}

Besides, in order to deal with the non-convex constraint \eqref{3st4}, let us define $\mathbf{R}_u=\mathbf{P}_u\mathbf{P}_u^H$ and $\mathbf{A}\left(\mathbf{o}_u,\mathbf{q}_0\right)=\mathbf{a}\left(\mathbf{o}_u,\mathbf{q}_0\right)\mathbf{a}\left(\mathbf{o}_u,\mathbf{q}_0\right)^H$ for notational convenience. Accordingly, we denote the entries in the $p$-th row and $q$-th column of $\mathbf{R}_u$ and $\mathbf{A}\left(\mathbf{o}_u,\mathbf{q}_0\right)$ as $\left[\mathbf{R}_u\right]_{p,q}$ and  $\left[\mathbf{A}\left(\mathbf{o}_u,\mathbf{q}_0\right)\right]_{p,q}$, respectively. The magnitude and phase of $\left[\mathbf{R}_u\right]_{p,q}$ are denoted by$\left|\left[\mathbf{R}_u\right]_{p,q}\right|$ and $\theta_{p, q}^{\mathbf{R}_u}$, respectively. Therefore, \eqref{3st4} can be rewritten equivalently as
\begin{align}\label{rss}
	\sum_{u \in \mathcal{U}} \frac{\operatorname{tr}\left(\mathbf{R}_u \mathbf{A}\left(\mathbf{o}_u, \mathbf{q}_0\right)\right)}{r^2\left(\mathbf{o}_u, \mathbf{q}_0\right)}  \geq \frac{ r^2\left(\mathbf{o}_0, \mathbf{q}_0\right) \sigma^2 \bar{\gamma}}{\left|\beta_0\right|}.
\end{align}
\begin{prop}\label{prop2}
	 We perform the first-order Taylor expansion on the left-hand-side (LHS) of \eqref{rss} at local point $\mathbf{o}_u^{(\upsilon)}$, then \eqref{rss} is reformulated as
\begin{equation}\label{3st4-1}
	H_u^{(\upsilon)}+{\mathbf{e}_u^{(\upsilon)}}^H\left(\mathbf{o}_u-\mathbf{o}_u^{(\upsilon)}\right) \geq \frac{ r^2\left(\mathbf{o}_0, \mathbf{q}_0\right) \sigma^2 \bar{\gamma}}{\left|\beta_0\right|},
\end{equation}
where
\begin{align}
	&H_u^{(\upsilon)}= \frac{\operatorname{tr}\left(\mathbf{R}_u \mathbf{A}\left(\mathbf{o}_u^{(\upsilon)}, \mathbf{q}_0\right)\right)}{r^2\left(\mathbf{o}_u^{(\upsilon)}, \mathbf{q}_0\right)},\\
	&\resizebox{0.98\hsize}{!}{$\begin{aligned}
			\mathbf{e}_u^{(\upsilon)}=
			\frac{\left[F\left(\mathbf{R}_u,\mathbf{o}_u^{(\upsilon)}, \mathbf{q}_0\right) r^2\left(\mathbf{o}_u^{(\upsilon)}, \mathbf{q}_0\right)
				-2\operatorname{tr}\left(\mathbf{R}_u \mathbf{A}\left(\mathbf{o}_u^{(\upsilon)}, \mathbf{q}_0\right)\right)\left(\mathbf{o}_u^{(\upsilon)}-\mathbf{q}_0\right)\right]}{ r^4\left(\mathbf{o}_u^{(\upsilon)}, \mathbf{q}_0\right)}	
		\end{aligned}$},\\
	&F\left(\mathbf{R}_u,\mathbf{o}_u^{(\upsilon)}, \mathbf{q}_0\right)=\left.\frac{\partial\operatorname{tr}\left(\mathbf{R}_u \mathbf{A}\left(\mathbf{o}_u, \mathbf{q}_0\right)\right)}{\partial \mathbf{o}_u}\right|_{\mathbf{o}_u=\mathbf{o}_u^{(\upsilon)}}.
\end{align} 

\end{prop}
\begin{proof}
	Please refer to Appendix \ref{proof2}.
\end{proof}

Furthermore, since $r^2\left(\mathbf{o}_i,\mathbf{q}_k\right)$ is a convex function with respect to $\mathbf{o}_i$, we replace the right-hand-side (RHS) of \eqref{r-ik} with its first-order Taylor expansion and re-expression \eqref{r-ik} as
\begin{equation}\label{r--ik}
	\resizebox{1\hsize}{!}{$\begin{aligned}
	\hat{r}_{i,k}\leq \left\|\mathbf{o}_i^{(\upsilon)}-\mathbf{q}_k\right\|^2+2\left(\mathbf{o}_i^{(\upsilon)}-\mathbf{q}_k\right)^T\left(\mathbf{o}_u-\mathbf{o}_u^{(\upsilon)}\right)+{H_i}^2,\forall k \in \mathcal{K}, k\notin \mathcal{ J }_i .
\end{aligned}$}
\end{equation} 

Finally, by replacing the non-convex constraints \eqref{3st2}, \eqref{3st3}, and \eqref{r-ik} as their approximated forms in \eqref{3st3-1}, \eqref{3st4-1}, and \eqref{r--ik}, we obtain the convex version of the optimization problem \text{(P2)} in the $\upsilon$-th iteration as
\begin{subequations}\label{P31}
	\begin{align}
		(\text{P2.}\upsilon):&\max _{\{f_{u,k}\}, \mathcal{O}^{(\upsilon)},\{\hat{r}_{i,k}\}}\quad \sum_{k \in \mathcal{K}}\mu_k f_{u,k}\\
		\mbox{s.t.}\quad
		&f_{u,k} \leq A_{u,k}^{(\upsilon)}+{\mathbf{c}_{u,k}^{(\upsilon)}}^H\left(\mathbf{o}_u-\mathbf{o}_u^{(\upsilon)}\right), \forall k \in \mathcal{K},k\in \mathcal{ J }_u,\\
		&\resizebox{0.86\hsize}{!}{$\begin{aligned}\left\|\mathbf{o}_u-\mathbf{q}_k\right\|^2\leq\frac{1}{\Psi}\left(\frac{\left|\mathds{1} \mathbf{p}_{ u, k,p}\right|^2}{2^{\frac{R_k^{t h}-C_k}{B}}-1}-\sum_{j \in \mathcal{ J }_u\backslash\{k\}}\left|\mathds{1}  \mathbf{p}_{u, j,p}\right|^2\right)-{H_u}^2,\forall k \in \mathcal{K},k\in \mathcal{ J }_u
			\end{aligned}$},\\
		&\resizebox{0.86\hsize}{!}{$\begin{aligned}H_u^{(\upsilon)}+{\mathbf{e}_u^{(\upsilon)}}^H\left(\mathbf{o}_u-\mathbf{o}_u^{(\upsilon)}\right) \geq \frac{ r^2\left(\mathbf{o}_0, \mathbf{q}_0\right) \sigma^2 \bar{\gamma}}{\left|\beta_0\right|},\forall u \in \mathcal{U}\end{aligned}$},\\
		&\resizebox{0.87\hsize}{!}{$\begin{aligned}
				\hat{r}_{i,k}\leq \left\|\mathbf{o}_i^{(\upsilon)}-\mathbf{q}_k\right\|^2+2\left(\mathbf{o}_i^{(\upsilon)}-\mathbf{q}_k\right)^T\left(\mathbf{o}_u-\mathbf{o}_u^{(\upsilon)}\right)+{H_i}^2,\forall k \in \mathcal{K}, k\notin \mathcal{ J }_i .
			\end{aligned}$}		
	\end{align}
\end{subequations}
 Hence, problem (\text{P2.}$\upsilon$) can be solved efficiently using the convex software tools such as CVX \cite{Grant2018}. The detailed process is given in Algorithm \ref{sca}.
\begin{algorithm}[t]
	\renewcommand{\algorithmicrequire}{\textbf{Inputs:}}
	\renewcommand{\algorithmicensure}{\textbf{Outputs:}}
	\caption{Iterative Optimization for Problem (\text{P2.}$\upsilon$)}
	\label{sca}
	\begin{algorithmic}[2]
		\REQUIRE{ $\mathcal{O}^0,\left\{\mathcal{J}_u^0\right\}_{u=1}^{U}, \left\{\mathbf{P}_u^0\right\}_{u=1}^{U}, \left\{C_k\right\}_{k=1}^{K}$. Set iteration number $\upsilon$,  }.
		\ENSURE  the set of UAV position $\mathcal{O}$. \\
		\STATE {\bf{begin:}}
		\STATE Set $ \upsilon=0 $.
		\STATE Set $\mathbf{o}_u^{(\upsilon)}=\mathbf{o}_u$.
		\REPEAT
		\STATE Solve convex problem (\text{P2.}$\upsilon$) with given UAV-CS association matrix, beamforming matrix and rate allocation vector for UAV location.
		\STATE Denote the optimal solution of (\text{P2.}$\upsilon$) by $ \mathbf{o}_u^{(\upsilon+1)} $,
		\STATE and set $ \upsilon = \upsilon + 1;$
		\UNTIL{the objective value (\text{P2.}$\upsilon$) converges.}
	\end{algorithmic}
\end{algorithm}
\subsection{Beamforming and Rate Allocation Optimization}\label{BeamformingandRateAllocationOptimization}
Now we optimize the UAV beamforming matrixes $\left\{\mathbf{P}_u\right\}_{u=1}^{U}$ and common rate allocation variables $\left\{C_k\right\}_{k=1}^{K}$ with the updated UAV-CS association cluster and the location of the UAV. By introducing $\hat R_k^p$ as the rate slack variable and letting $\mathbf{P}_{u,c}=\mathbf{p}_{u,c}\mathbf{p}_{u,c}^H$, $\mathbf{P}_{u, k,p}=\mathbf{p}_{u, k,p}\mathbf{p}_{u, k,p}^H$, $\mathbf{P}_{u,r}= \mathbf{p}_{u,r} \mathbf{p}_{u,r}^H$, $\mathbf{H}_{u, k}=\mathbf{h}_{u, k}\mathbf{h}_{u, k}^H$, the optimization problem can be rewritten as
\begin{subequations}\label{P4}
	\begin{align}
		&(\text{P3}):\max _{ \left\{\mathbf{P}_{u,c}\right\}_{u=1}^{U},\left\{\mathbf{P}_{u, k,p}\right\}_{k=1}^{K},\atop \left\{\mathbf{P}_{u,r}\right\}_{u=1}^{U} \left\{C_k\right\}_{k=1}^{K},\left\{\hat R_k^p\right\}_{k=1}^{K} } \sum_{k \in \mathcal{K}} \mu_k\left(C_k+\hat R_k^p\right)\\
		&\mbox{s.t.}\quad
		 \label{4st0}C_k\geq\text{0},\forall k\in\mathcal{K},\\
		\label{4st1}&\resizebox{0.95\hsize}{!}{$\begin{aligned} B\log _2\left(1+\frac{\sum_{u \in \mathcal{U}}\operatorname{tr}\left(\mathbf{H}_{u, k}\mathbf{P}_{u,c}\right)}{ \sum_{u \in \mathcal{U}} \sum_{j \in \mathcal{J}_u} \operatorname{tr}\left(\mathbf{H}_{u, k}\mathbf{P}_{u, j,p}\right)+\sigma^2}\right) \geq \sum_{i \in \mathcal{K}} C_{i},\forall k \in \mathcal{K}
	 	\end{aligned}$}, \\
		\label{4st2}&\resizebox{0.95\hsize}{!}{$\begin{aligned} B\log _2\left(1+\frac{\operatorname{tr}\left(\mathbf{H}_{u, k}\mathbf{P}_{u,k,p}\right)}{\sum_{i \in \mathcal{U} } \sum_{j \in \mathcal{J}_i \backslash \{k\}}\operatorname{tr}\left(\mathbf{H}_{i, k}\mathbf{P}_{i, j,p}\right)+\sigma^2}\right) \geq \hat R_k^p,\forall k \in \mathcal{K}
		\end{aligned}$} ,\\
		\label{4st3}& \operatorname{tr}\left(\mathbf{P}_{u,c}+ \sum_{k=1}^K\mathbf{P}_{u, k,p}+\mathbf{P}_{u,r}\right) \leq P_{max},\forall u\in\mathcal{U},\\
		\label{4st4}&\mathbf{P}_{u,c} \succeq 0, \mathbf{P}_{u, k,p} \succeq 0, \mathbf{P}_{u,r} \succeq 0, \forall u \in \mathcal{U},\forall k \in \mathcal{K}, \\
		\label{4st5}&
		\resizebox{0.95\hsize}{!}{$\begin{aligned} \operatorname{rank}\left(\mathbf{P}_{u,c}\right)=1, \operatorname{rank}\left(\mathbf{P}_{u, k,p}\right)=1,\operatorname{rank}\left(\mathbf{P}_{u,r}\right)=1, \forall u \in \mathcal{U}, \forall k \in \mathcal{K}
			\end{aligned}$} , 
		 \\
		\label{4st6}&	\resizebox{0.95\hsize}{!}{$\begin{aligned}\sum_{u \in \mathcal{U}} \frac{\operatorname{tr}\left(\left(\mathbf{P}_{u,c}+ \sum_{k=1}^K\mathbf{P}_{u, k,p}+\mathbf{P}_{u,r}\right) \mathbf{A}\left(\mathbf{o}_u, \mathbf{q}_0\right)\right)}{r^2\left(\mathbf{o}_u, \mathbf{q}_0\right)}  \geq \frac{ r^2\left(\mathbf{o}_0, \mathbf{q}_0\right) \sigma^2 \bar{\gamma}}{\left|\beta_0\right|}	\end{aligned}$},\\
		\label{4st7}&C_k+\hat R_k^p\geq R_k^{th}\text{,}\forall k\in\mathcal{K}.
	\end{align}
\end{subequations}
Problem (\text{P3}) is an SDP problem, in which the constraints \eqref{4st1}, \eqref{4st2}, and \eqref{4st5} are non-convex.  

To solve problem (P3), we first introduce slack variables $\left\{\eta_{ k}\right\}_{k=1}^K$ and $\left\{\rho_{ k}\right\}_{k=1}^K$ \cite{Yin2022}, then the constraint \eqref{4st1} can be transformed as
\begin{align}
	\label{eta-rhoc}
	&\eta_{ k}-\rho_{k} \geq \sum_{i \in \mathcal{K}} C_i \frac{\log 2}{B} , \forall k \in \mathcal{K},\\
	\label{eta-c}
	&\resizebox{1\hsize}{!}{$\begin{aligned}
			e^{\eta_{ k}} \leq \sum_{u \in \mathcal{U}}\operatorname{tr}\left(\mathbf{H}_{u, k}\mathbf{P}_{u,c}\right)+ \sum_{u \in \mathcal{U}} \sum_{j \in \mathcal{J}_u} \operatorname{tr}\left(\mathbf{H}_{u, k}\mathbf{P}_{u, j,p}\right)+\sigma^2, \forall k \in \mathcal{K},
		\end{aligned}$}\\
	\label{rho-c}
	&e^{\rho_{ k}} \geq \sum_{u \in \mathcal{U}} \sum_{j \in \mathcal{J}_u} \operatorname{tr}\left(\mathbf{H}_{u, k}\mathbf{P}_{u, j,p}\right)+\sigma^2, \forall k \in \mathcal{K}.
\end{align}
Similarly, by introducing slack variables $\left\{\chi_k\right\}_{k=1}^K,\left\{\zeta_k\right\}_{k=1}^K$,  \eqref{4st2} can be rewritten as
\begin{align}
	\label{chi-zeta}
	&\chi_k-\zeta_k \geq \hat R_k^p \frac{\log 2}{B}, \forall k \in \mathcal{K},\\
	\label{chi}
	&e^{\chi_k} \leq  \sum_{i \in \mathcal{U} } \sum_{j \in \mathcal{J}_i } \operatorname{tr}\left(\mathbf{H}_{i, k}\mathbf{P}_{i, j,p}\right)+\sigma^2, \forall k \in \mathcal{K},\\
	\label{zeta}
	 &e^{\zeta_k} \geq \sum_{i \in \mathcal{U} } \sum_{j \in \mathcal{J}_i \backslash \{k\}}\operatorname{tr}\left(\mathbf{H}_{i, k}\mathbf{P}_{i, j,p}\right)+\sigma^2, \forall k \in \mathcal{K}.
\end{align}
It can be observed that the constraints \eqref{rho-c} and \eqref{zeta} are still non-convex. To deal with it, at the $\kappa$-th step, we first approximate the LHS of \eqref{rho-c} based on their first-order Taylor expansion at local point $\rho_{ k}^{(\kappa)}$, and accordingly re-express \eqref{rho-c} as
\begin{equation}\label{Taylor-rhoc}
	\resizebox{1\hsize}{!}{$\begin{aligned}
	e^{\rho_{k}^{(\kappa)}}\left(1+\rho_{k}-\rho_{ k}^{(\kappa)}\right)\geq \sum_{u \in \mathcal{U}} \sum_{j \in \mathcal{J}_u} \operatorname{tr}\left(\mathbf{H}_{i, k}\mathbf{P}_{i, j,p}\right)+\sigma^2, \forall k \in \mathcal{K}.
\end{aligned}$}
\end{equation}
Similarly, \eqref{zeta} can be re-expressed as
\begin{align}\label{Taylor-zeta}
	\resizebox{1\hsize}{!}{$\begin{aligned}
			e^{\zeta_{ k}^{(\kappa)}}\left(1+\zeta_{ k}-\zeta_{ k}^{(\kappa)}\right) \geq \sum_{i \in \mathcal{U} } \sum_{j \in \mathcal{J}_i \backslash \{k\}}\operatorname{tr}\left(\mathbf{H}_{i, k}\mathbf{P}_{i, j,p}\right)+\sigma^2, \forall k \in \mathcal{K}.
		\end{aligned}$}
\end{align}
Next, we deal with the non-convex rank constraints in \eqref{4st5} via the idea of SDR. In particular, we relax the rank constraints in \eqref{4st5} and (\text{P3}) can be reformulated as
	\begin{align}
	\label{PSDR4}
		(\text{P3.SDR}):&\max _{\substack{\left\{\mathbf{P}_{u,c}\right\}_{u=1}^{U},\left\{\mathbf{P}_{u, k,p}\right\}_{k=1}^{K},\\ \left\{\mathbf{P}_{u,r}\right\}_{u=1}^{U} \left\{C_k\right\}_{k=1}^{K}, \left\{\hat R_k^p\right\}_{k=1}^{K},\\\left\{\eta_{ k}\right\}_{k=1}^K,\left\{\rho_{ k}\right\}_{k=1}^K,\left\{\chi_{ k}\right\}_{k=1}^K,\left\{\zeta_{ k}\right\}_{k=1}^K} } \sum_{k \in \mathcal{K}} \mu_k\left(C_k+\hat R_k^p\right)\\
		\text{s.t.}\quad 
		&\eqref{4st0},\eqref{4st6},\eqref{4st7},\eqref{eta-rhoc},\eqref{eta-c},\eqref{chi-zeta},\nonumber\\&\eqref{chi},\eqref{Taylor-rhoc},\eqref{Taylor-zeta}\nonumber
	\end{align}
Note that problem (\text{P3. SDR}) is a convex SDP and thus can be solved by CVX. It is worth noting that  the ranks of solution $\left\{\mathbf{P}_{u,c}\right\}_{u=1}^{U},\left\{\mathbf{P}_{u, k,p}\right\}_{k=1}^{K}, \left\{\mathbf{P}_{u,r}\right\}_{u=1}^{U}$ to (\text{P3}) are equal to 1, and can be obtained via eigenvalue decomposition (EVD). Otherwise, the Gaussian randomization method is utilized to construct sub-optimal rank-one solution \cite{ZLSDR2010}. 

\subsection{Complexity Analysis}\label{ComplexityAnalysis}
In this subsection, we analyze the complexity of the proposed algorithms. The complexity for solving problem \text{(P0)} is mainly determined by the complexity of solving problem \text{(P1)}, \text{(P2)}, and (\text{P3}). Specifically, problem \text{(P1)} is solved by K-Means algorithm, and the complexity of solving problem \text{(P1)} is $\mathcal{O}\left(2KU\epsilon_1\right)$ \cite{MacQueen1967}, where $\epsilon_1$ is the accuracy of the K-Means algorithm for solving problem \text{(P1)}. Next, problem \text{(P2)} is solved by the SCA technique, and transformed as problem (\text{P2.}$\upsilon$) at each iteration. There are $(K+1)(U+1)-1$ constraints in problem (\text{P2.}$\upsilon$), so the number of iterations that are required for the SCA technique is $\mathcal{O}\left(\sqrt{(K+1)(U+1)-1} \log _2\left(1 / \epsilon_2\right)\right)$ \cite{Grant2018}, where $\epsilon_2$ is the accuracy of the SCA technique for solving problem (\text{P2.}$\upsilon$). At each iteration, the complexity of solving problem (\text{P2.}$\upsilon$) is $\mathcal{O}\left(S_1^2 S_2\right)$, where $S_1=U(K+1)$ and $S_2=(K+1)(U+1)-1$ are the total numbers of the variables and constraints, respectively. Thus, the total complexity of the SCA technique for
solving problem \text{(P2)} is $\mathcal{O}\left(U^{3.5} K^{3.5} \log _2\left(1 / \epsilon_2\right)\right)$. In a similar manner, the computational complexity of problem (\text{P3})
is $\mathcal{O}\left({(N_t^2UK)}^{3.5}\log _2\left(1 / \epsilon_3\right)\right)$ \cite{Yin2022}, where $\epsilon_3$ is the accuracy of the SCA technique for solving problem (\text{P3}). As a result, the total complexity for solving problem \text{(P0)} is $\resizebox{1\hsize}{!}{$\begin{aligned}
		\mathcal{O}\left(2KU\epsilon_1+T U^{3.5} K^{3.5} \log _2\left(1 / \epsilon_2\right)+T{(N_t^2UK)}^{3.5} \log _2\left(1 / \epsilon_3\right)\right)
	\end{aligned}$}$ \cite{YLi2021}, where $T$ is the number of iterations for the proposed CoRSMA-ISAC scheme.

\section{Simulation Results}\label{SimulationResults}
In this section, we provide simulations to demonstrate the effectiveness of the proposed CoRSMA-ISAC scheme. We consider a 2D disaster area of $500\ \text{m}\times 500\ \text{m} $ where TS is located at the center of the area, i.e., $\mathbf{q}_0=(250,250)$, and CSs uniformly distributed in such an area. All UAVs are assumed to fly at a fixed altitude $H_u=H_0=100\ \text{m}$ with $N_t=8$ antenna elements. Unless otherwise stated, the maximum transmit power of each UAV is fixed at $P_{max}=25\ \text{dBm}$, and the reference channel power is $\varepsilon_0= -60\ \text{dB}$. The reference sensing channel power is $\beta_0=-50\ \text{dB}$. The noise power is assumed to be $\sigma^2=-110\ \text{dBm}$ \cite{WL2018} and the channel bandwidth is $B=1\ \text{MHz}$. Finally, the threshold for the total achievable rate of each CS is set as $R_k^{th}=1 \ \text{Mbps}$, the sensing SNR threshold is set as $\bar\gamma=2$ \cite{demirhan2023cellfree}, and the rate weight is assumed to $\mu_k=1/K, \forall k\in \mathcal{K}$. 
%
%
\begin{figure*}[!htb]
	\centering
	\subfigure[UAV-CS association and UAV initial deployment]
	{
		\begin{minipage}[b]{0.32\linewidth}
			\centering
			\includegraphics[width=1.95in]{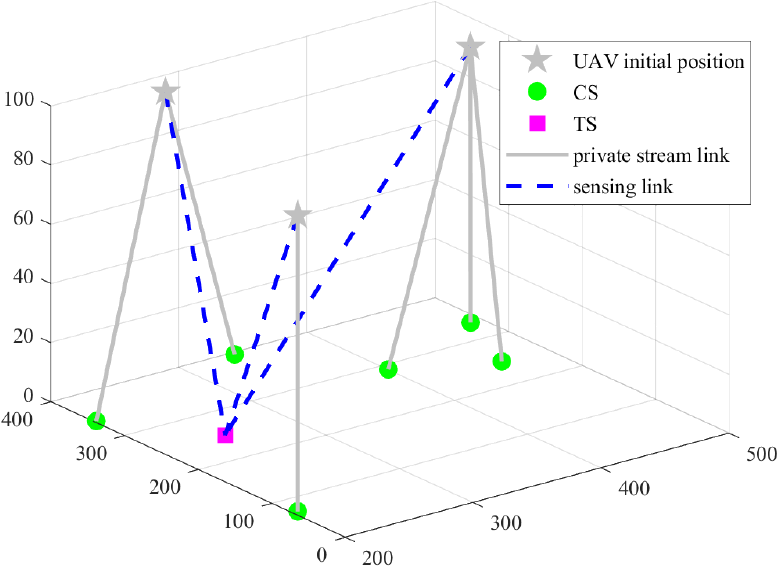}
			\label{kmeansdeploy}
		\end{minipage}%
	}
	\subfigure[ UAV optimized deployment]
	{
		\begin{minipage}[b]{0.32\linewidth}
			\centering
			\includegraphics[width=2in]{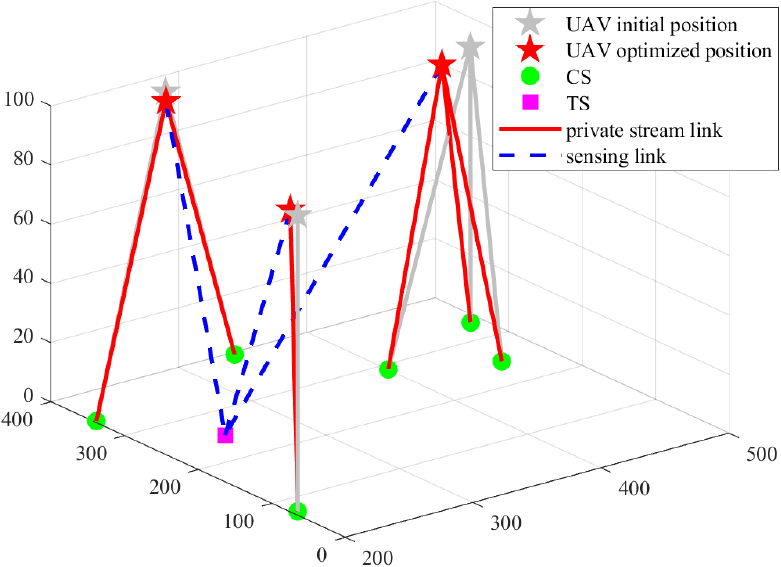}
			\label{deployment}
		\end{minipage}
	}
	\subfigure[Top view of UAV deployment process]
	{
		\begin{minipage}[b]{0.3\linewidth}
			\centering
			\includegraphics[width=1.65in]{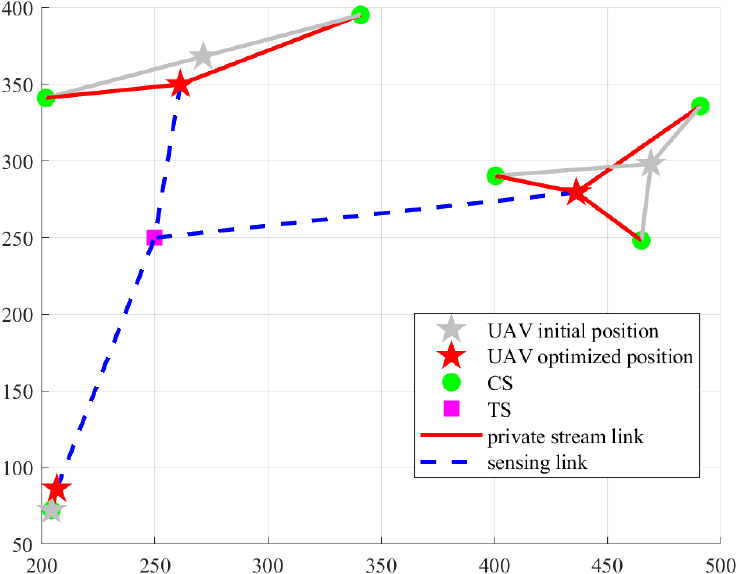}
			\label{deployment2D}
		\end{minipage}
	}
	\caption{The procedure of UAV deployment and UAV-CS association, where $U = 3$ and $K = 6$ are fixed.}
	\label{deployments}
\end{figure*}

\begin{figure}[h]
	\center
	\includegraphics[width=3.2in]{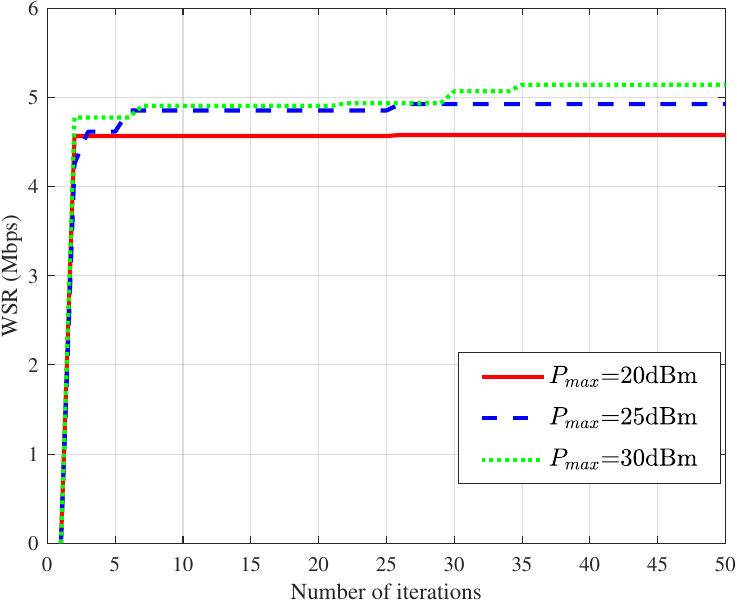}\\
	\caption{The convergence behaviour of the proposed CoRSMA-ISAC scheme with different $P_{max}$, where $U=3$ and $K=5$ are fixed.}\label{convergence}
\end{figure}

Fig. \ref{deployments} shows the procedure of UAV-CS association and UAV deployment by our proposed CoRSMA-ISAC scheme. The number of CSs is set to $K=6$. In Fig. \ref{kmeansdeploy}, it can be observed that all UAVs first obtain the UAV-CS association and the initial UAV deployment based on \eqref{kmeans-J} and \eqref{kmeans-o} by K-Means algorithm. Next, Fig. \ref{deployment} depicts that with the initialized UAV deployment, the optimal UAV deployment is finally obtained after several iterations of the proposed CoRSMA-ISAC scheme. Fig. \ref{deployment2D} shows a top view of the variation of UAV deployment, from which it can be seen that compared to the initial position of UAVs, their optimal position is getting closer to each other and surrounding the TS compared. This is because that each UAV transmits the common message to serve not only its associated CSs but also other unassociated CSs, while meeting the requirement of sensing SNR.

Next, the convergence behavior of the proposed CoRSMA-ISAC scheme is investigated in Fig. \ref{convergence}. We can clearly see that the algorithm can converge quickly to reach the maximum WSR after serveral iterations. Besides, it is expected that more transmit power budget $P_{max}$ we can utilize, a larger value of WSR is obtained. 



As a benchmark comparison, we consider the traditional ISAC design with SDMA. The corresponding communication rate in \eqref{Rp} can be expressed as
\begin{equation}\label{Rp_sdma}
	R_k^{\text{SDMA}}=B\log _2\left(1+\frac{\left|\mathbf{h}_{u, k}^H \mathbf{p}_{u, k}\right|^2}{ \sum_{i \in \mathcal{U}} \sum_{j \in \mathcal{J}_i \backslash \{k\}} \left|\mathbf{h}_{i, k}^H \mathbf{p}_{ i, j}\right|^2+\sigma^2}\right).
\end{equation}

We also consider the traditional ISAC design with NOMA where each ISAC UAV divides the available spectrum equally and communicates with its associated CSs via NOMA. According to the NOMA scheme, SIC is employed to decode the inter-CS interference. Without loss of generality, we reset the CS associated with the $u$th UAV index as $\mathcal{J}_u=\{k_{u-1}+1, \cdots, k_u\}$, with $k_{0}=1$, $k_{U}=K$. Simultaneously, we assume that $\left\|\mathbf{h}_{u, k_{u-1}+1}\right\|<\cdots<\left\|\mathbf{h}_{u, k_u}\right\|$. Thus, the rate of the CS $k\in \{k_{u-1}+1, \cdots, k_u-1\}$ can be expressed as
\begin{equation}\label{Rp_noma1}
	R_k^{\text{NOMA}}=\frac{B}{U}\cdot\log _2\left(1+\frac{\left|\mathbf{h}_{u, k}^H \mathbf{p}_{u, k}\right|^2}{ \sum_{i=k+1}^{ k_u} \left|\mathbf{h}_{u, k}^H \mathbf{p}_{u, i}\right|^2+\sigma^2}\right).
\end{equation}
For the CS $k_u,\forall u\in \mathcal{U}$, the rate is given by
\begin{equation}\label{Rp_noma2}
	R_{k_u}^{\text{NOMA}}=\frac{B}{U}\cdot\log _2\left(1+\frac{\left|\mathbf{h}_{u, k_u}^H \mathbf{p}_{u, k_u}\right|^2}{\sigma^2}\right).
\end{equation}

Moreover, we also consider the traditional ISAC design with OMA where CSs use separate spectrum and divide the spectrum equally for transmission. The corresponding communication rate in \eqref{Rp} can be expressed as
\begin{equation}\label{Rp_oma}
	R_k^{\text{OMA}}=\frac{B}{K}\cdot\log _2\left(1+\frac{\left|\mathbf{h}_{u, k}^H \mathbf{p}_{u, k}\right|^2}{\sigma^2}\right).
\end{equation}

\begin{figure}[t]
	\center
	\includegraphics[width=3.2in]{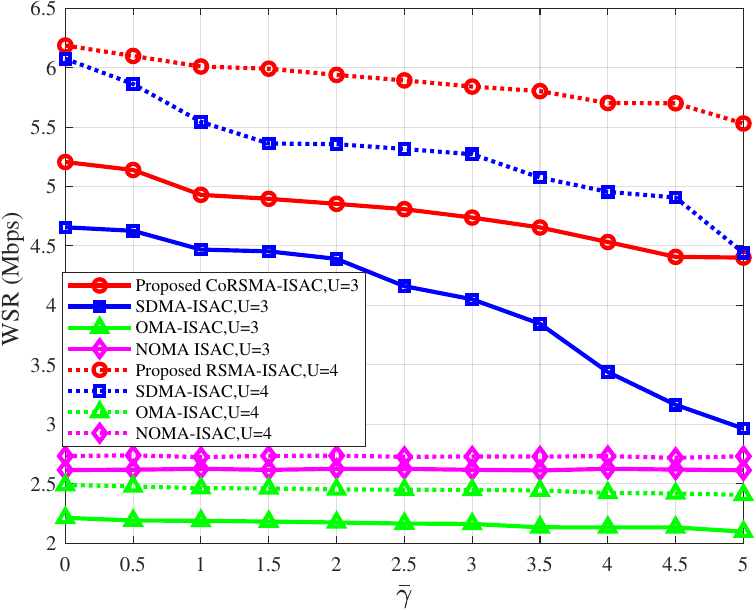}\\
	\caption{WSR with different sensing SNR requirement $\bar{\gamma}$, where $K=5$ is fixed. }\label{CompareRateRs}
\end{figure}
\begin{figure}[t]
	\center
	\includegraphics[width=3.15in]{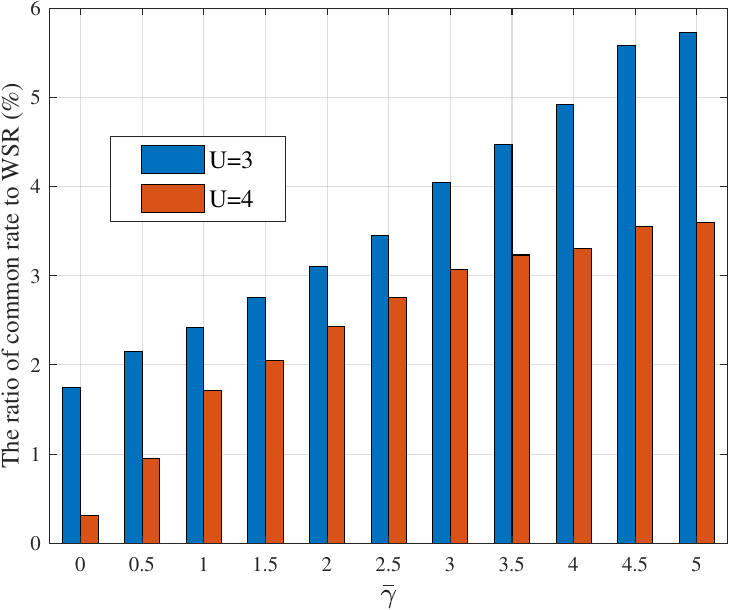}\\
	\caption{The ratio of common rate to WSR in CoRSMA-ISAC scheme with different sensing SNR requirement $\bar{\gamma}$, where $K=5$ is fixed.}\label{rate_c_with_rs}
\end{figure}
It can be seen in Fig. \ref{CompareRateRs} that under the same sensing SNR constraint, if there are more UAVs transmit common message, the WSR will increase since performance of the interference management is enhanced with limited spectrum. Furthermore, Fig. \ref{CompareRateRs} illustrates a trade-off between communication and sensing performance. It can be observed that with the increase of sensing requirements, the WSR decreases. It is expected because when the sensing requirement is higher, more power is required to be allocated for sensing, then the power used for communication is relatively reduced and the WSR decreases. Besides, as $\bar\gamma$ increases, CoRSMA-ISAC can achieve better communication performance compared to SDMA-ISAC, NOMA-ISAC, and OMA-ISAC. This is because that sensing also benifits from the beamforming gain of common message transmission by different UAVs with CoRSMA-ISAC. Compared to CoRSMA-ISAC and SDMA-ISAC, NOMA-ISAC performs worse due to the inefficient use of SIC layers in multi-antenna NOMA, which lowers the sum DoF compared to CoRSMA-ISAC and SDMA-ISAC \cite{BC2021}. Specifically, due to the rate threshold requirement, NOMA-ISAC decreases the power allocated to the CS with better channel gain in order to ensure fairness for CSs with worse channel gain, resulting in a worse WSR. On the other hand, because of non-coordination of UAV, UAVs transmit with orthogonal spectrum, then the reduction of transmission spectrum resources further leads to a decrease in WSR. Fig. \ref{rate_c_with_rs} shows that the ratio of common rate to WSR increases as  $\bar\gamma$ increases, which also verifies that the common message transmission will help sensing in the proposed CoRSMA-ISAC. Morever, when the number of UAV increases, CoRSMA-ISAC reduces the ratio of common rate to achieve a higher WSR. This is because that UAV-CS assoication is changed, each CS can choose more proper UAV for private message transmission.


\begin{figure}[t]
	\center
	\includegraphics[width=3.2in]{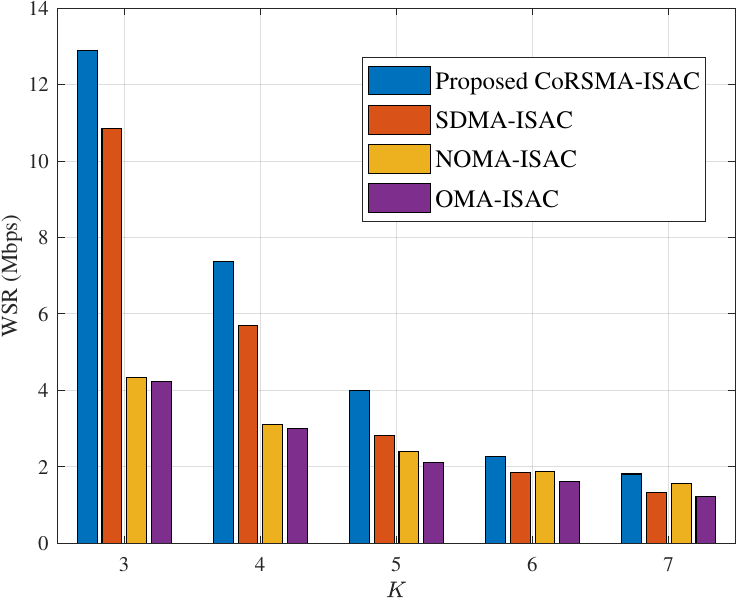}\\
	\caption{WSR versus different number of CSs $K$ , where $U=3$ is fixed. }\label{CSchange}
\end{figure}
\begin{figure}[t]
	\center
	\includegraphics[width=3.2in]{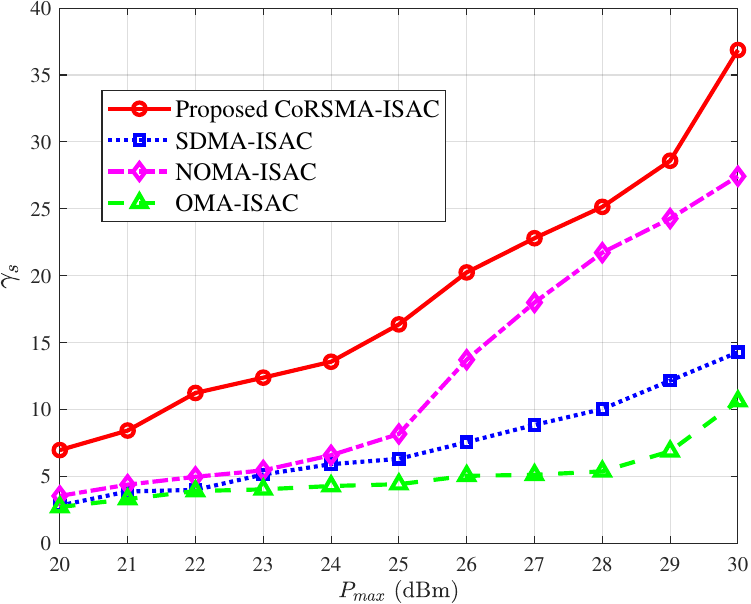}\\
	\caption{Sensing SNR with different maximum transmit power of each UAV $P_{max}$, where $U=3$ and $K=5$ are fixed. }\label{RsWithPtotal}
\end{figure}
On the other hand, we compare WSR with different number of CSs. As shown in Fig. \ref{CSchange}, as the number of CSs increases, more CSs will use the same spectrum for transmission, and  WSR will decreases due to the multi-user intra and intercell interference. However, compared to other schemes, CoRSMA-ISAC with flexible decoding rules can mitigate the decrease of WSR caused by spectrum congestion by increasing the power allocation of common message.

\begin{figure}[t]
	\center
	\includegraphics[width=3.2in]{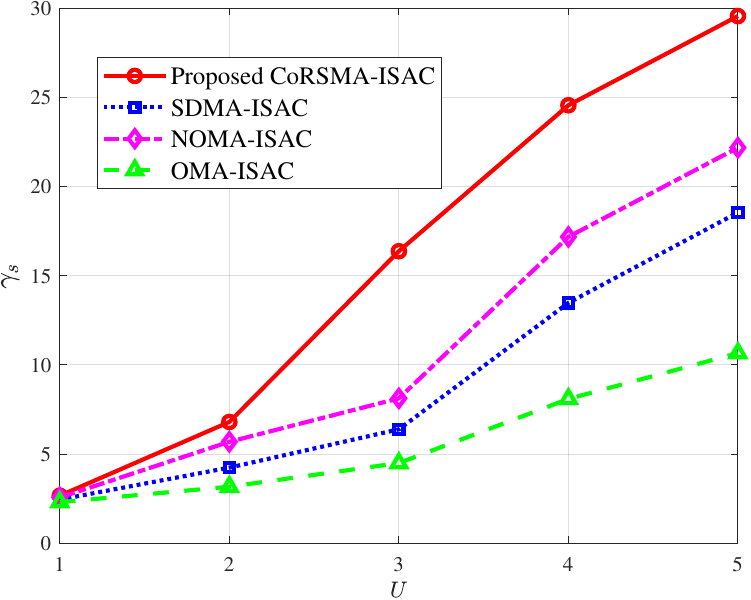}\\
	\caption{Sensing SNR with different number of ISAC UAV $U$ where $K=5$ is fixed. }\label{RswithUAV_line}
\end{figure}
We further discuss the sensing performance of the proposed CoRSMA-ISAC scheme. From Fig. \ref{RsWithPtotal}, it can be seen that CoRSMA-ISAC also achieve a higher sensing SNR with the same $P_{max}$ compared to OMA-ISAC, NOMA-ISAC and SDMA-ISAC. Specifically, when the $P_{max}$ is large enough, such an improvement is more significant. Similar to Fig. \ref{rate_c_with_rs}, it is still due to the benefits of common message transmission on sensing. Besides, Fig. \ref{RswithUAV_line} shows that cooperative sensing by multiple UAVs can enhance the sensing SNR compared to simple single UAV sensing ($U=1$). Such a sensing SNR enhancement is more significant when the number of UAVs is larger.

\section{Conclusion}\label{Conclusion}
In this paper, we have investigated a joint UAV-CS association, UAV deployment, and beamforming optimization problem for a coordinated RSMA for ISAC in emergency UAV system. We presented the communication and sensing signal model, and formulated an optimization problem to maximize the WSR of the system while satisfying the sensing SNR constraint. To solve the formulated non-convex optimization problems, we proposed an efficient algorithm based on K-Means, SCA, and SDR techniques. The numerical results indicated that the proposed CoRSMA-ISAC scheme can achieve higher WSR while satisfying sensing SNR constraints in spectrum congestion. Moreover, the proposed CoRSMA-ISAC scheme also guarantees a higher sensing SNR under the same maximum transmit power constraint. In the future work, we will consider more practical scenarios in the emergency events, such as the mobility of communication users and targets, the existence of clutters, 3D UAV deployment, etc.


\appendix
\subsection{Derivation of Sensing SNR in \eqref{rs}}\label{proof-rs}
First, the numerator of \eqref{rs} can be simplified as \cite{demirhan2023cellfree}
\begin{align}\label{E-GS}
	&\resizebox{1\hsize}{!}{$\begin{aligned}
				\mathbb{E}\left[\|\mathbf{G} \mathbf{S}\|_{\mathrm{F}}^2\right]=\text{tr}\left\{\mathbb{E}\left[\mathbf{S}\mathbf{S}^H\mathbf{G}^H\mathbf{G}\right]\right\}=\text{tr}\left\{\underbrace{\mathbb{E}\left[\mathbf{S}\mathbf{S}^H\right]}_{L\cdot\mathbf{I}} \mathbb{E} \left[\mathbf{G}^H\mathbf{G}\right]\right\}
		\end{aligned}$}		\nonumber\\
	&\resizebox{0.89\hsize}{!}{$\begin{aligned}=L\sum_{u \in \mathcal{U}}\left( \text{tr}\left\{\mathbb{E} \left[\bar{\mathbf{G}}_u^H\bar{\mathbf{G}}_u\right]\right\} + \sum_{u' \in \mathcal{U}\backslash\{u\}}\text{tr}\left\{\mathbb{E} \left[\bar{\mathbf{G}}_u^H\bar{\mathbf{G}}_{u'}\right]\right\}\right)
		\end{aligned}$}	
\end{align}
Due to the expectation over $\beta_u$ and independence of these random variables, we have $\mathbb{E} \left[\bar{\mathbf{G}}_u^H\bar{\mathbf{G}}_{u'}\right]=0$. Further, we have
\begin{align}\label{E-GG}
	&\text{tr}\left\{\mathbb{E} \left[\bar{\mathbf{G}}_u^H\bar{\mathbf{G}}_u\right]\right\}=\nonumber\\
		&	\resizebox{1\hsize}{!}{$\begin{aligned}\text{tr}\left\{\mathbb{E} \left[\left|\beta_u\right|^2 \mathbf{a}\left(\mathbf{o}_u,\mathbf{q}_0\right)\mathbf{a}^H\left(\mathbf{o}_u,\mathbf{q}_0\right)\mathbf{P}_u\mathbf{P}_u^H \underbrace{\mathbf{b}\left(\mathbf{o}_0,\mathbf{q}_0\right) \mathbf{b}^H\left(\mathbf{o}_0,\mathbf{q}_0\right) }_{N_r}\right]\right\}
	\end{aligned}$}\nonumber\\
&=\frac{\left|\beta_0\right|N_r}{r^2\left(\mathbf{o}_u,\mathbf{q}_0\right)r^2\left(\mathbf{o}_0,\mathbf{q}_0\right) }\|\mathbf{a}^H\left(\mathbf{o}_u,\mathbf{q}_0\right) \mathbf{P}_u\|^2.                            
\end{align}
Similarly, we have 
\begin{equation}\label{E-N}
	\mathbb{E}\left[\|\mathbf{N}\|_{\mathrm{F}}^2\right]=L N_r \sigma^2.
\end{equation}
With \eqref{E-GS}, \eqref{E-GG}, and \eqref{E-N}, \eqref{rs} can be obtained. This completes the proof.
\subsection{Proof of Proposition \ref{prop1}} \label{proof1}
By following \eqref{rp} and \eqref{st6}, we have \eqref{prof-rth} at the top of the next page,
\begin{figure*}
	\begin{equation}\label{prof-rth}
		\frac{\left|\mathds{1}\mathbf{p}_{u,k,p}\right|^2}{\sum_{j\in\mathcal{J}_u\setminus\{k\}}\left|\mathds{1}\mathbf{p}_{u,j,p}\right|^2+\sum_{i\in\mathcal{U}\setminus\{u\}}\sum_{j\in\mathcal{J}_i}\frac{r^2\left(\mathbf{o}_u,\mathbf{q}_k\right)}{r^2\left(\mathbf{o}_i,\mathbf{q}_k\right)}\left|\mathds{1}\mathbf{p}_{i,j,p}\right|^2+\sigma^2 \varepsilon_0^{-1}r^2\left(\mathbf{o}_u,\mathbf{q}_k\right)}\geq2^{R_k^{th}-C_k}-1,
	\end{equation}
	\hrulefill
\end{figure*}
which is further arranged to get
\begin{align}\label{pro1-r}
	&\frac{\left|\mathds{1}\mathbf{p}_{u,k,p}\right|^2}{2^{\frac{R_k^{th}-C_k}{B}}-1}-\sum_{j\in\mathcal{J}_u\setminus\{k\}}\lvert\mathds{1}\mathbf{p}_{u,j,p}\rvert^2\geq\nonumber\\ &r^2\left(\mathbf{o}_u,\mathbf{q}_k\right)(\sum_{i\in\mathcal{U}\setminus\{u\}}\sum_{j\in\mathcal{J}_i}\frac1{r^2\left(\mathbf{o}_i,\mathbf{q}_k\right)}\lvert\mathds{1}\mathbf{p}_{i,j,p}\rvert^2+\sigma^2\varepsilon_0^{-1}).
\end{align}
Since \eqref{pro1-r} is still non-convex with respect to $r^2\left(\mathbf{o}_u,\mathbf{q}_k\right)$, we use slack variables  $\{\hat{r}_{i,k}\leq r^2\left(\mathbf{o}_i,\mathbf{q}_k\right),\forall k \in \mathcal{K}, k\notin \mathcal{ J }_i \}$ to replace $r^2\left(\mathbf{o}_i,\mathbf{q}_k\right)$ in \eqref{pro1-r}. Then we have
\begin{align}\label{proof1-1}
	&\frac{\left|\mathds{1}\mathbf{p}_{u,k,p}\right|^2}{2^{\frac{R_k^{th}-C_k}{B}}-1}-\sum_{j\in\mathcal{J}_u\setminus\{k\}}\lvert\mathds{1}\mathbf{p}_{u,j,p}\rvert^2\geq\nonumber \\&r^2\left(\mathbf{o}_u,\mathbf{q}_k\right)(\underbrace{\sum_{i\in\mathcal{U}\setminus\{u\}}\sum_{j\in\mathcal{J}_i}\frac1{\hat{r}_{i,k}}\lvert\mathds{1}\mathbf{p}_{i,j,p}\rvert^2+\sigma^2\varepsilon_0^{-1}}_{\Psi}).
\end{align}
Notice that $r^2\left(\mathbf{o}_u,\mathbf{q}_k\right)=\left\|\mathbf{o}_u-\mathbf{q}_k\right\|^2+{H_u}^2$, then \eqref{proof1-1} can be re-expressed as 
\begin{align}
	r^2\left(\mathbf{o}_u,\mathbf{q}_k\right)&=\left\|\mathbf{o}_u-\mathbf{q}_k\right\|^2+{H_u}^2\nonumber\\
	&\leq \frac{1}{\Psi}\left(\frac{\left|\mathds{1} \mathbf{p}_{u, k,p}\right|^2}{2^{\frac{R_k^{t h}-C_k}{B}}-1}-\sum_{j \in \mathcal{J}_u\setminus\{k\}}\left|\mathds{1}  \mathbf{p}_{u, j,p}\right|^2\right).
\end{align}
Further, we have the convex constraint expressed as
\begin{align}
	\resizebox{0.89\hsize}{!}{$\begin{aligned}
	\left\|\mathbf{o}_u-\mathbf{q}_k\right\|^2 \leq \frac{1}{\Psi}\left(\frac{\left|\mathds{1} \mathbf{p}_{u, k,p}\right|^2}{2^{\frac{R_k^{t h}-C_k}{B}}-1}-\sum_{j \in \mathcal{J}_u\setminus\{k\}}\left|\mathds{1}  \mathbf{p}_{u, j,p}\right|^2\right)-{H_u}^2.
\end{aligned}$}
\end{align}
This completes the proof.
\subsection{Proof of Proposition \ref{prop2}}\label{proof2}
The entry in the $p$-th row and $q$-th column of $\mathbf{A}\left(\mathbf{o}_u, \mathbf{q}_0\right)$ can by expressed as
\begin{equation}\label{Apq}
	\left[\mathbf{A}\left(\mathbf{o}_u, \mathbf{q}_0\right)\right]_{p, q}=e^{j \pi (p-q) \frac{H_u}{r\left(\mathbf{o}_u, \mathbf{q}_0\right)}}.
\end{equation}
It is observed from \eqref{Apq} that $\mathbf{R}_u $ and $\mathbf{A}\left(\mathbf{o}_u, \mathbf{q}_0\right)$ are Hermitian matrices \cite{Lyu2023jointmaneuver}, and thus we have 
\begin{align}
	\label{tr-PA}
	\operatorname{tr}\left(\mathbf{R}_u \mathbf{A}\left(\mathbf{o}_u, \mathbf{q}_0\right)\right)
	 =&\sum_{p=1}^{N_t} \sum_{q=1}^{N_t}\left[\mathbf{R}_u\right]_{p, q} \mathrm{e}^{j \pi (q-p)\frac{H_u}{r\left(\mathbf{o}_u, \mathbf{q}_0\right)}} \nonumber\\
	 =&\sum_{a=1}^{N_t}\left[\mathbf{R}_u\right]_{a, a}+2 \sum_{p=1}^{N_t} \sum_{q=p+1}^{N_t}\left|\left[\mathbf{R}_u\right]_{p, q}\right|\nonumber\\
	  &\times \cos \left(\theta_{p, q}^{\mathbf{R}_u}+ \pi (q-p) \frac{H_u}{r\left(\mathbf{o}_u, \mathbf{q}_0\right)}\right).
\end{align}
Then the first-order derivative of \eqref{tr-PA} with respect to $\mathbf{o}_u$ is derived as
\begin{align}\label{tr-PA-dao}
	&F\left(\mathbf{R}_u,\mathbf{o}_u, \mathbf{q}_0\right)=\frac{\partial \operatorname{tr}\left(\mathbf{R}_u \mathbf{A}\left(\mathbf{o}_u, \mathbf{q}_0\right)\right)}{\partial \mathbf{o}_u}\nonumber\\
	=&\resizebox{0.95\hsize}{!}{$\begin{aligned}2 \pi \sum_{p=1}^{N_t} \sum_{q=p+1}^{N_t}\left|\left[\mathbf{R}_u\right]_{p, q}\right|
	\times\sin \left[\theta_{p, q}^{\mathbf{R}_u}+ \pi (q-p) \frac{H_u}{r\left(\mathbf{o}_u,\mathbf{q}_0\right)}\right]\end{aligned}$}\nonumber\\
	&\times \frac{(q-p)H_u \left(\mathbf{o}_u-\mathbf{q}_0\right)}{r^3\left(\mathbf{o}_u,\mathbf{q}_0\right)}.
\end{align}
Then $F\left(\mathbf{R}_u,\mathbf{o}_u^{(\upsilon)}, \mathbf{q}_0\right)=\left.\frac{\partial\operatorname{tr}\left(\mathbf{R}_u \mathbf{A}\left(\mathbf{o}_u, \mathbf{q}_0\right)\right)}{\partial \mathbf{o}_u}\right|_{\mathbf{o}_u=\mathbf{o}_u^{(\upsilon)}}$.

With \eqref{tr-PA} and \eqref{tr-PA-dao}, the LHS of \eqref{rss} based on its first-order Taylor expansion at local point $\mathbf{o}_u^{(\upsilon)}$ can be approximated as
\begin{equation}
	\sum_{u \in \mathcal{U}}\left(H_u^{(\upsilon)}+{\mathbf{e}_u^{(\upsilon)}}^H\left(\mathbf{o}_u-\mathbf{o}_u^{(\upsilon)}\right)\right) ,
\end{equation}
where
\begin{align}
	H_u^{(\upsilon)}= \frac{\operatorname{tr}\left(\mathbf{R}_u \mathbf{A}\left(\mathbf{o}_u^{(\upsilon)}, \mathbf{q}_0\right)\right)}{r^2\left(\mathbf{o}_u^{(\upsilon)}, \mathbf{q}_0\right)},
\end{align}
\begin{align}
	\resizebox{1\hsize}{!}{$\begin{aligned}
	\mathbf{e}_u^{(\upsilon)}=
	\frac{\left[F\left(\mathbf{R}_u,\mathbf{o}_u^{(\upsilon)}, \mathbf{q}_0\right) r^2\left(\mathbf{o}_u^{(\upsilon)}, \mathbf{q}_0\right)
		-2\operatorname{tr}\left(\mathbf{R}_u \mathbf{A}\left(\mathbf{o}_u^{(\upsilon)}, \mathbf{q}_0\right)\right)\left(\mathbf{o}_u^{(\upsilon)}-\mathbf{q}_0\right)\right]}{ r^4\left(\mathbf{o}_u^{(\upsilon)}, \mathbf{q}_0\right)}.	
\end{aligned}$}
\end{align}
This completes the proof.

\end{document}